\documentclass[
    10pt,
    twocolumn,
    prd,
    preprintnumbers,
    superscriptaddress,
    altaffilletter,
    nofootinbib,
    aps
]{revtex4-2}

\usepackage[a4paper,margin=1in]{geometry}
\usepackage[colorlinks,citecolor=blue,urlcolor=red,bookmarks=false,hypertexnames=true]{hyperref}
\usepackage{amsmath}
\usepackage{amssymb}
\usepackage{amsfonts}
\usepackage{amsthm}
\usepackage{graphicx}
\usepackage{lineno}
\usepackage[usenames,dvipsnames]{xcolor}
\usepackage{xspace}
\usepackage{caption}
\usepackage{subcaption}
\usepackage{acronym}
\usepackage{tikz}
\usetikzlibrary{arrows,shapes,trees,decorations.pathreplacing}
\usepackage{orcidlink}
\usepackage{enumitem}
\usepackage{empheq}

\newtheorem{definition}{Definition}[section]
\newtheorem{theorem}[definition]{Theorem}
\newtheorem{lemma}[definition]{Lemma}

\begin{document}
\title{Dual Cutler–Vallisneri Corrections: Mitigating PSD Drift in Zero-Latency Gravitational-Wave Searches}

\author{James Kennington$^*$ \orcidlink{0000-0002-6899-3833}}
\affiliation{Department of Physics, The Pennsylvania State University, University Park, PA 16802, USA}
\affiliation{Institute for Gravitation and the Cosmos, The Pennsylvania State University, University Park, PA 16802, USA}

\date{\today}


\begin{abstract}
Maximizing pre-merger warning times in gravitational-wave searches requires minimizing algorithmic latency. 
While current pipelines typically rely on truncated linear-phase filters, \emph{minimum-phase whitening} offers a zero-latency alternative that eliminates the acausal look-ahead buffer. 
However, this causal approach exposes the analysis to \emph{spectral drift}, where the whitening operator applied to live data diverges from the static template bank, creating a functional perturbation of the matched-filter metric. 
We develop a perturbative framework generalizing the Cutler--Vallisneri formalism to address these \emph{metric errors}, deriving analytic expressions for the resulting timing, phase, and SNR biases. 
Validated against exact stationary-phase models and numerical injections, these corrections achieve $<1\%$ error. 
Applying this framework to GWTC-4.0 events with realistic 1-week power spectral density (PSD) lags, we find that uncorrected drift induces severe systematics: detector-pair timing biases exceeding $200 \mu$s, phase shifts up to 0.2 rad, and sky-localization errors of $5^{\circ}-10^{\circ}$.
Additionally, we observe a median signal-to-noise ratio (SNR) loss of $3-5\%$, with outliers exceeding $8\%$.
These results demonstrate that while minimum-phase whitening maximizes the early-warning window, analytic drift corrections are essential to maintain detection volume and pointing accuracy in future observing runs.
\end{abstract}
\maketitle
\section{Introduction}
Modern gravitational-wave (GW) data-analysis pipelines rely on matched filtering to extract weak astrophysical signals from non-stationary, colored detector noise \cite{theligoscientificcollaborationGW150914FirstResults2016,allenFINDCHIRPAlgorithmDetection2012a,cutlerGravitationalWavesMergin1994a}. 
The optimal filter requires whitening the detector strain time series so that its residual noise becomes spectrally flat, ensuring the matched filter achieves maximum signal-to-noise ratio (SNR) under Gaussian assumptions \cite{wainsteinExtractionSignalsNoise1970,littenbergBayesLineBayesianInference2015}. 
Traditionally, whitening is performed using linear-phase filters, which perfectly preserve the phase of the input data but incur a fixed group delay \cite{oppenheimDiscretetimeSignalProcessing2010,smithIntroductionDigitalFilters2008}. 
While negligible in offline searches, this latency is problematic in low-latency pipelines such as GstLAL and SGNL, where every second matters for prompt electromagnetic follow-up \cite{messickAnalysisFrameworkPrompt2017b,cannonEarlyWarningDetectionGravitational2012d,nitzPyCBCLiveRapid2018,adamsLowlatencyAnalysisPipeline2016a}.

To maximally reduce whitening-related latency while preserving fidelity, minimum-phase (MP) whitening can be used, in which the filter’s zeros lie strictly inside the unit circle, guaranteeing causality and minimizing group delay \cite{oppenheimDiscretetimeSignalProcessing2010,smithIntroductionDigitalFilters2008}. 
The trade-off is a frequency-dependent phase rotation that must be characterized and controlled. 
A prototype zero-latency whitening (ZLW) implementation based on discrete-time spectral factorization and the folded-cepstrum method showed that such filters can achieve near-real-time operation without numerical instability \cite{tsukadaApplicationZerolatencyWhitening2018}. 
However, no major low-latency search pipelines currently deploy minimum-phase whitening in production \cite{ewingPerformanceLowlatencyGstLAL2024,nitzRapidDetectionGravitational2018,chuSPIIROnlineCoherent2021,alleneMBTAPipelineDetecting2025}. 
The analytic framework developed in this work provides the theoretical foundation required to enable and validate such an implementation in future releases.

In realistic operation, the detector noise power spectral density (PSD) evolves over time due to instrumental configuration, control-loop behavior, and environmental coupling \cite{buikemaSensitivityPerformanceAdvanced2020}. 
Continuous detector-characterization studies document this non-stationarity alongside the appearance and evolution of narrow spectral features (“lines”) across observing runs \cite{covasIdentificationMitigationNarrow2018}. 
Since the PSD used to whiten the detector data is updated in real time while the template bank is typically generated and whitened using a fixed reference PSD, the two can gradually diverge \cite{mageeImpactSelectionBiases2024}. 
This implies that the effective noise weighting applied to the data and the templates is no longer perfectly matched, introducing a functional perturbation to the matched-filter metric. 
While locally small, these spectral divergences integrate across the sensitive bandwidth to produce systematic biases in arrival time and phase that can exceed the statistical uncertainty of high-SNR events \cite{littenbergSeparatingGravitationalWave2010}. 
Uncorrected, these whitening inconsistencies degrade timing precision and pointing accuracy, threatening the efficacy of low-latency electromagnetic follow-up.

We formalize this problem geometrically. 
Standard perturbative approaches, such as the Cutler--Vallisneri (CV) formalism~\cite{cutlerLISADetectionsMassive2007}, typically address \emph{vector perturbations}, where a waveform model error $\delta h$ displaces the signal vector off the search manifold under a fixed metric. 
In this work, we address the dual problem: \emph{metric perturbations}, where the signal model is fixed, but the inner product structure itself warps due to spectral drift. 
We first derive a generalized CV theorem applicable to any functional metric deformation, providing the theoretical basis for recovering the exact projection of these errors onto the extrinsic parameter space~\cite{owenSearchTemplatesGravitational1996a}.

We then apply this general framework to the specific case of minimum-phase whitening, the requisite technique for zero-latency operation~\cite{oppenheimDiscretetimeSignalProcessing2010, tsukadaApplicationZerolatencyWhitening2018}. 
By treating the PSD drift as a generator of whitening-phase distortion, we obtain the \emph{Dual CV corrections}, the explicit analytic expressions for the resulting timing and phase biases. 
We validate these results analytically using a controlled Newtonian stationary-phase approximation (SPA) model~\cite{cutlerGravitationalWavesMergin1994a, allenFINDCHIRPAlgorithmDetection2012a}, where the matched-filter integrals admit closed-form solutions. Further, we validate these results numerically through pipeline-agnostic matched-filter injections. 
Finally, we apply the framework to real detector data from the GWTC-4.0 catalog~\cite{collaborationGWTC40UpdatingGravitationalWave2025,collaborationOpenDataLIGO2025} to assess the astrophysical impact on timing and sky localization under realistic operational lags. 
By quantifying the scale and geometry of whitening-induced biases, this framework provides the foundation for \emph{safe, zero-latency operation} in next-generation pipelines~\cite{messickAnalysisFrameworkPrompt2017b, nitzPyCBCLiveRapid2018, chuSPIIROnlineCoherent2021}, informing both whitening refresh cadence and real-time phase-compensation strategies.
\section{Motivation}

The geometric framework of matched filtering has long been central to gravitational-wave data analysis, primarily for constructing efficient template banks over the intrinsic mass--spin manifold~\cite{owenSearchTemplatesGravitational1996a, hannaBinaryTreeApproach2022c}. However, the extrinsic geometry of the signal-to-noise ratio (SNR) surface~\cite{allenFINDCHIRPAlgorithmDetection2012a, cutlerGravitationalWavesMergin1994a} is equally susceptible to curvature deformations when the data and template whitening filters diverge~\cite{tsukadaApplicationZerolatencyWhitening2018, mageeImpactSelectionBiases2024}. In low-latency pipelines~\cite{messickAnalysisFrameworkPrompt2017b, ewingPerformanceLowlatencyGstLAL2024, nitzPyCBCLiveRapid2018, alleneMBTAPipelineDetecting2025, chuSPIIROnlineCoherent2021}, such deformations manifest as systematic timing and phase biases that degrade network coincidence and sky localization~\cite{fairhurstTriangulationGravitationalWave2009, sullivanTimingSystemLIGO2023, fairhurstLocalizationTransientGravitational2018}. This work unifies these perspectives by treating spectral drift as a perturbative force acting on the extrinsic search subspace. By extending metric-based methods to capture these dynamical whitening effects, we derive analytic correction factors—paralleling treatments of other PSD variations~\cite{zackayDetectingGravitationalWaves2021, littenbergSeparatingGravitationalWave2010}—that allow pipelines to recover the true signal geometry without the prohibitive latency cost of re-filtering.
\section{Methods}
\label{sec:methods}

In this section, we formalize the problem of minimum-phase whitening under spectral drift. We define the perturbative framework for PSD evolution and apply the geometric theorems derived in Appendix~\ref{app:geom} to obtain explicit, first-order analytic corrections for the resulting timing and phase biases.

\subsection{PSD Perturbations and Mismatch}
\label{subsec:mpmp-setup}

We restrict our analysis to a single detector; network-level effects are incorporated later via triangulation geometry (Sec.~\ref{subsubsec:catalog-validation}). Let $S_1(f)$ be a fixed reference PSD, and let $W_1(f)$ denote its minimum-phase whitener, normalized such that $|W_1(f)|^2 = 1/S_1(f)$. In an idealized matched-filter search, both the frequency-domain data $d(f)$ and the template family $h_{(t,\phi)}(f)$ are whitened with $W_1$. The complex matched filter output is then
\begin{equation}
  \rho_0(t,\phi) = \bigl\langle W_1 d,\; W_1 h_{(t,\phi)} \bigr\rangle,
  \label{eq:rho0-def}
\end{equation}
where $\langle u, v\rangle = \int_{-\infty}^{\infty} u(f)\,v^{*}(f)\,df$ denotes the standard $L^2$ inner product~\cite{wainsteinExtractionSignalsNoise1970,allenFINDCHIRPAlgorithmDetection2012a}. Under these ideal conditions, the true coalescence parameters $(t_0,\phi_0)$ form a nondegenerate local maximizer of the squared-SNR surface
\begin{equation}
  J_0(t,\phi) := |\rho_0(t,\phi)|^2.
\end{equation}
The Hessian of $J_0$ at the peak determines the local sensitivity geometry and is proportional to the Fisher information metric restricted to the $(t,\phi)$ subspace (see Appendix~\ref{app:geom})~\cite{cutlerGravitationalWavesMergin1994a,owenSearchTemplatesGravitational1996a}.

In realistic low-latency operations, the templates are whitened with a static reference kernel $W_1$ (fixed during bank generation), while the live data streams are whitened with a \emph{dynamic} minimum-phase filter $W_2$ derived from the locally estimated PSD $S_2(f)$~\cite{adamsLowlatencyAnalysisPipeline2016a,nitzPyCBCLiveRapid2018,nitzRapidDetectionGravitational2018,chuSPIIROnlineCoherent2021,ewingPerformanceLowlatencyGstLAL2024,tsukadaApplicationZerolatencyWhitening2018}. This creates a mismatched statistic:
\begin{equation}
  \rho_\epsilon(t,\phi) = \bigl\langle W_2 d,\; W_1 h_{(t,\phi)} \bigr\rangle,
  \label{eq:rho-mpmp-def}
\end{equation}
where the subscript $\epsilon$ anticipates the perturbative regime $|S_2 - S_1| \ll S_1$. We treat $\rho_\epsilon$ as a one-parameter family of statistics and track the shift of the maximizer of $J_\epsilon = |\rho_\epsilon|^2$ as the mismatch grows~\cite{biscoveanuQuantifyingEffectPower2020,vitaleEffectCalibrationErrors2012a,driggersImprovingAstrophysicalParameter2019}.

To model non-stationary spectral evolution~\cite{buikemaSensitivityPerformanceAdvanced2020,covasIdentificationMitigationNarrow2018,davisImprovingSensitivityAdvanced2019}, we write the data PSD as a functional perturbation of the reference:
\begin{equation}
  S_2(f) = S_1(f)\,e^{2\epsilon a(f)}, \quad 0 < \epsilon \ll 1,
  \label{eq:psd-perturb}
\end{equation}
where $a(f)$ is a real, bounded, smooth drift function. The dimensionless parameter $\epsilon$ controls the perturbation scale, satisfying $\epsilon\|a\|_{L^\infty(\mathcal{B})} \ll 1$ over the analysis band $\mathcal{B}$. This exponential form conveniently mirrors the log-amplitude structure of minimum-phase factorization (Appendix~\ref{app:mp}). It captures common drift morphologies, such as localized line features (where $a(f)$ is a Gaussian bump) or broadband sensitivity shifts (where $a(f)$ varies slowly)~\cite{buikemaSensitivityPerformanceAdvanced2020,covasIdentificationMitigationNarrow2018,davisImprovingSensitivityAdvanced2019,biscoveanuQuantifyingEffectPower2020}.

We relate this spectral drift to the induced perturbation of the whitening kernel. Minimum-phase whitening associates a unique complex frequency response to any PSD $S_i$:
\begin{equation}
  W_i(f) = \exp\!\bigl[L_i(f) + i\phi_i(f)\bigr], \quad L_i(f) = -\tfrac{1}{2}\ln S_i(f),
  \label{eq:Wi-logform}
\end{equation}
where the phase $\phi_i = \mathcal{H}[L_i]$ is the Hilbert transform of the log-magnitude~\cite{oppenheimDiscretetimeSignalProcessing2010,smithIntroductionDigitalFilters2008,bocheSpectralFactorizationWhitening2005,damera-venkataDesignOptimalMinimumphase2000,mecozziNecessarySufficientCondition2016,tsukadaApplicationZerolatencyWhitening2018}. Substituting Eq.~\eqref{eq:psd-perturb} into Eq.~\eqref{eq:Wi-logform}, the log-magnitude and phase perturb linearly:
\begin{equation}
\begin{aligned}
  L_2(f) &= L_1(f) - \epsilon a(f), \\
  \phi_2(f) &= \phi_1(f) - \epsilon \mathcal{H}[a](f).
\end{aligned}
  \label{eq:Lphi-perturb}
\end{equation}
Consequently, the relative perturbation of the whitening kernel is given to first order by:
\begin{equation}
  \frac{\delta W(f)}{W_1(f)} \equiv \frac{W_2(f)}{W_1(f)} - 1 \approx -\epsilon \bigl[a(f) + i\Phi_a(f)\bigr],
  \label{eq:deltaW}
\end{equation}
where we have defined the \emph{whitening phase perturbation} $\Phi_a(f) \equiv \mathcal{H}[a](f)$. The uniform bound $\|\delta W/W_1\| = \mathcal{O}(\epsilon)$ ensures validity of the expansion~\cite{ephremidzeNonoptimalSpectralFactorizations2016,ephremidzeQuantitativeResultsContinuity2020}.

To isolate the geometric impact of this mismatch, we work in the signal-only limit~\cite{cutlerGravitationalWavesMergin1994a,vitaleEffectCalibrationErrors2012a,biscoveanuQuantifyingEffectPower2020,mageeImpactSelectionBiases2024}, setting $d(f) = h_{(t_0,\phi_0)}(f)$. In the next subsection, we apply the Generalized Cutler--Vallisneri Theorem (Appendix~\ref{app:geom}) to the perturbed surface $J_\epsilon$ to derive the analytic timing and phase biases $(\delta t, \delta\phi)$.
\subsection{First-Order Timing and Phase Shifts}
\label{subsec:first-order-mpmp}

We now apply the geometric result of Appendix~\ref{app:geom} to the MP--MP configuration of Sec.~\ref{subsec:mpmp-setup} to obtain explicit first-order expressions for the timing and phase biases.

Recall that for each whitening mismatch we have a matched-filter family $J_\epsilon(t,\phi) := |\rho_\epsilon(t,\phi)|^2$, with $(t_0,\phi_0)$ a non-degenerate local maximizer of $J_0$. Writing $\theta^i = (t,\phi)$, we view $J_\epsilon(\theta)$ as a smooth one-parameter family of scalar fields on the two-dimensional parameter space. The shift-of-maximum corollary (Theorem~\ref{thm:gen-cv} in Appendix~\ref{app:geom}) states that the first-order shift $\delta\theta$ satisfies~\cite{cutlerGravitationalWavesMergin1994a,vitaleEffectCalibrationErrors2012a}:
\begin{equation}
\begin{aligned}
  \delta\theta^i &= g^{ij}\,k_j, \\
  g_{ij} &\equiv -\partial_i\partial_j J_0(\theta_0), \\
  k_j &\equiv \partial_\epsilon \partial_j J_\epsilon(\theta_0)\big|_{\epsilon=0},
\end{aligned}
  \label{eq:geom-shift-J}
\end{equation}
where $g_{ij}$ is the Hessian metric (proportional to Fisher information) and $k_j$ is the generalized force vector.

We first compute the metric components $g_{ij}$. Adopting the standard extrinsic parametrization~\cite{cutlerGravitationalWavesMergin1994a}:
\begin{equation}
  h_{(t,\phi)}(f) = h_{\mathrm{intr}}(f)\,e^{i(2\pi f t - \phi)},
  \label{eq:extrinsic-factorisation}
\end{equation}
we define the whitened signal power $w(f) \equiv |W_1(f)h_{\mathrm{intr}}(f)|^2$. Working in the signal-only gauge $d = h_{(t_0,\phi_0)}$ and shifting coordinates so $(t_0,\phi_0) = (0,0)$~\cite{cutlerGravitationalWavesMergin1994a,vitaleEffectCalibrationErrors2012a,biscoveanuQuantifyingEffectPower2020}, the unperturbed surface is
\begin{equation}
  \rho_0(t,\phi) = \int w(f)\,e^{i(2\pi f t - \phi)}\,df, \quad J_0 = |\rho_0|^2.
\end{equation}
Evaluating the Hessian at the peak yields the diagonal metric components:
\begin{equation}
\begin{aligned}
  g_{tt} &= -\partial_t^2 J_0(0,0) = (2\pi)^2 \int f^2\,w(f)\,df, \\
  g_{t\phi} &= -\partial_t\partial_\phi J_0(0,0) = 0, \\
  g_{\phi\phi} &= -\partial_\phi^2 J_0(0,0) = \int w(f)\,df.
\end{aligned}
  \label{eq:metric-components}
\end{equation}
These components encode the local curvature (stiffness) of the SNR peak against timing and phase shifts~\cite{owenSearchTemplatesGravitational1996a,cutlerGravitationalWavesMergin1994a}.

Next, we compute the force vector $k_j$. From Sec.~\ref{subsec:mpmp-setup}, the first-order whitening perturbation is
\begin{equation}
  \frac{\delta W}{W_1} \approx -\epsilon\bigl[a(f) + i\Phi_a(f)\bigr].
\end{equation}
Differentiating $J_\epsilon = |\rho_\epsilon|^2$ with respect to $\epsilon$ at the peak shows that only the \emph{imaginary} part of the perturbation (phase) contributes to the shift; the real part (amplitude drift $a(f)$) affects the peak height but not its location to first order. The resulting force components are:
\begin{equation}
\begin{aligned}
  k_t &= \partial_\epsilon \partial_t J_\epsilon(\theta_0)\big|_{\epsilon=0} = (2\pi)\int f\,\Phi_a(f)\,w(f)\,df, \\
  k_\phi &= \partial_\epsilon \partial_\phi J_\epsilon(\theta_0)\big|_{\epsilon=0} = \int \Phi_a(f)\,w(f)\,df.
\end{aligned}
  \label{eq:forcing-components}
\end{equation}
Geometrically, the whitening phase error $\Phi_a(f)$ acts as a frequency-dependent potential that tilts the SNR manifold.

Finally, combining the metric and force terms (and noting $g_{t\phi}=0$), we obtain the \textbf{Dual Cutler--Vallisneri corrections}:
\begin{align}
  \delta t^{(1)} &= g^{tt}k_t = \frac{\int (2\pi f)\,\Phi_a(f)\,w(f)\,df}{\int (2\pi f)^2\,w(f)\,df},
  \label{eq:dt-first-order} \\[0.5em]
  \delta\phi^{(1)} &= g^{\phi\phi}k_\phi = \frac{\int \Phi_a(f)\,w(f)\,df}{\int w(f)\,df}.
  \label{eq:dphi-first-order}
\end{align}
These formulas represent the main analytic result: timing and phase biases are power-weighted spectral averages of the whitening phase mismatch. The denominators represent the signal's stiffness (bandwidth), while the numerators represent the driving force of the spectral drift. In Sec.~\ref{sec:results}, we validate these expressions against exact models and numerical injections.
\subsection{Validation Strategy}
\label{subsec:validation-setups}

We validate the first-order Dual CV corrections using three complementary frameworks, each probing a distinct regime of the theory: exact analytic consistency, numerical robustness under controlled injection, and operational relevance using real astrophysical data.

\subsubsection*{Analytic SPA Validation}
For theoretical cross-validation, we employ a Newtonian stationary-phase approximation (SPA) inspiral model ~\cite{drozGravitationalWavesInspiraling1999, allenFINDCHIRPAlgorithmDetection2012a} ($|\tilde{h}(f)| \propto f^{-7/6}$) embedded in a power-law noise environment $S_1(f) \propto f^p$. In this idealized setting, the whitened signal power $w(f)$ admits a closed analytic form, allowing us to derive the \emph{exact} shift of the SNR peak under a linear whitening-phase perturbation $\Phi_a(f) = \alpha + \beta f$ (see Appendix~\ref{app:spa-linear}). We compare this exact solution $(\Delta t_{\mathrm{full}}, \Delta\phi_{\mathrm{full}})$ directly against the perturbative prediction $(\delta t^{(1)}, \delta\phi^{(1)})$ derived from our geometric theorem. This comparison provides a stringent test of the mathematical formalism, verifying that the first-order expansion captures the true physical displacement of the matched-filter maximum in the limit of small drift.

\subsubsection*{Numerical PSD-Drift Validation}
To test the corrections under realistic signal-processing conditions, we implement a full pipeline-agnostic injection study using the \texttt{zlw} codebase. We utilize a reference PSD $S_1(f)$ and a perturbed PSD $S_2(f) = S_1(f)\exp[2\epsilon a(f)]$, where $a(f)$ is a Gaussian bump modeling a drifting spectral feature (e.g., a wandering calibration line centered at $150\,\mathrm{Hz}$ with width $40\,\mathrm{Hz}$). The signal model consists of SPA waveforms injected into a $T=4\,\mathrm{s}$ analysis segment with a sampling rate of $4096\,\mathrm{Hz}$ ($N_{\mathrm{FFT}}=16384$ points), covering a frequency band of $20$--$512\,\mathrm{Hz}$. Minimum-phase whitening filters $W_1$ and $W_2$ are constructed via the folded-cepstrum algorithm (Appendix~\ref{app:mp}) ~\cite{oppenheimDiscretetimeSignalProcessing2010, tsukadaApplicationZerolatencyWhitening2018} on this grid. We evaluate the matched-filter surfaces $\rho_0 = \langle W_1 d, W_1 h \rangle$ and $\rho_\epsilon = \langle W_2 d, W_1 h \rangle$, extracting the ``true'' numerical biases $(\Delta t_{\mathrm{full}}, \Delta\phi_{\mathrm{full}})$ by maximizing the interpolated SNR$^2$ surface $J_\epsilon(t,\phi)$ within a $\pm 5\,\mathrm{ms}$ window. By varying the drift amplitude $\epsilon$ over three orders of magnitude ($10^{-4} \le \epsilon \le 10^{-1}$), we quantify the convergence rate of our analytic approximation and identify the breakdown scale where higher-order metric terms become significant.

\subsubsection*{Heuristic Catalog Study}
\label{subsubsec:methods-catalog}
To assess the impact on real gravitational-wave detections, we perform a re-analysis of compact binary coalescences from the GWTC-4.0 catalog~\cite{collaborationGWTC40UpdatingGravitationalWave2025}. We select events with confident source-frame masses and GPS times, converting parameters to the detector frame. For each contributing interferometer (H1, L1, V1) \cite{theligoscientificcollaborationProspectsObservingLocalizing2020, theligoscientificcollaborationGuideLIGOVirgoDetector2020}, we construct a \textbf{spectral drift pair} using open strain data obtained from the Gravitational Wave Open Science Center (GWOSC)~\cite{collaborationOpenDataLIGO2025}. The ``live'' data PSD ($S_2$) is estimated from $128\,\mathrm{s}$ of strain centered on the event time $t_c$ using median-Welch averaging with $4\,\mathrm{s}$ segments, matching the typical scale of low-latency tracking. The reference PSD ($S_1$) is estimated from a science segment approximately \textbf{1 week prior} to the event using longer $32\,\mathrm{s}$ segments, simulating the stale, high-precision whitening filter often used in template bank generation. Using \textsc{LALSimulation}, we generate the corresponding IMRPhenomD waveform \cite{husaFrequencydomainGravitationalWaves2016} $\tilde{h}(f)$ on a $20$--$1024\,\mathrm{Hz}$ grid and compute the per-detector Dual CV corrections $(\delta t^{(1)}, \delta\phi^{(1)})$ and fractional SNR loss $\delta\rho/\rho$. Finally, to estimate the impact on multimessenger follow-up, we propagate these timing biases into sky-localization errors. We interpret the shifts $\delta t^{(1)}$ as additive errors in the arrival time difference $\Delta \tau_{ij}$ between detector pairs and solve for the perturbed sky position $(\alpha', \delta')$ that minimizes the residual delays relative to the catalog location. This yields the induced angular shift $\Delta\Omega$ and the systematic offset in Right Ascension and Declination $(\Delta\alpha, \Delta\delta)$, providing a direct measure of the pointing error introduced by uncorrected PSD drift.
\section{Results}
\label{sec:results}

This section presents three converging lines of evidence validating the Dual Cutler--Vallisneri framework: analytic verification in a controlled SPA model, numerical injection tests with realistic minimum-phase whitening, and an empirical assessment of biases in GWTC-4.0 events. Our approach parallels the foundational work of Cutler and Vallisneri~\cite{cutlerLISADetectionsMassive2007}, which characterized how vector perturbations of the \emph{signal model} (waveform errors) induce parameter shifts via the likelihood geometry. We generalize this picture to the dual regime essential for low-latency operations: rather than perturbing the waveform vector, we consider functional perturbations of the \emph{inner product metric} itself, driven by the spectral evolution of the whitening filter. By treating coalescence time and phase $(t, \phi)$ on the same geometric footing as intrinsic parameters, the resulting formulas quantify how minimum-phase whitening drift deforms the extrinsic SNR surface. In what follows, we demonstrate that these analytic corrections accurately predict the systematic biases observed in practical detection pipelines.
\subsection{Analytical Validation via SPA Model}
\label{sec:results-spa}

We begin by validating the first-order Dual CV corrections using a controlled Newtonian stationary-phase approximation (SPA) inspiral model \cite{drozGravitationalWavesInspiraling1999,cutlerGravitationalWavesMergin1994a}. In the case of a strictly linear whitening-phase perturbation $\Phi_a(f) = \alpha + \beta f$, the matched-filter integral admits a closed analytic form in terms of incomplete gamma functions (see Appendix~\ref{app:spa-linear}). The \emph{exact} first-order shifts of the SNR peak derived from this closed form are:
\begin{equation}
  \Delta\phi_{\mathrm{full}}^{(1)} = \epsilon\left[\alpha + \beta\frac{m_1}{m_0}\right], \quad
  \Delta t_{\mathrm{full}}^{(1)} = \frac{\epsilon}{2\pi}\left[\beta + \alpha\frac{m_1}{m_2}\right],
\end{equation}
where $m_k = \int f^{k-n}df$ are the SPA frequency moments \cite{poissonGravityNewtonianPostNewtonian2014a}. These expressions coincide \emph{exactly} with the geometric predictions $(\delta t^{(1)}, \delta\phi^{(1)})$ derived in Sec.~\ref{subsec:first-order-mpmp}. This term-by-term identity provides a stringent internal consistency check, confirming that the perturbative geometric derivation recovers the true physical displacement of the maximum in an analytically tractable regime.

To probe a more realistic setting with nonlinear phase structure, we extend the perturbation to quadratic order: $\Phi_a(f) = \alpha + \beta f + \gamma f^2$. While no closed form exists for this case, the full matched filter $\rho_\epsilon$ can be evaluated numerically with high precision. We compare the exact numerical shift against the geometric prediction:
\begin{equation}
  \delta\phi^{(1)} = \epsilon\,\frac{\int w(f)\Phi_a(f)\,df}{\int w(f)\,df}.
\end{equation}
As shown in Fig.~\ref{fig:spa-psd}, the analytic correction tracks the full numerical bias $\Delta\phi_{\mathrm{full}}$ across more than two decades of perturbation strength $\epsilon$. The fractional error $|\delta\phi^{(1)} - \Delta\phi_{\mathrm{full}}|/|\Delta\phi_{\mathrm{full}}|$ scales linearly with $\epsilon$, exhibiting clean first-order convergence ($\mathcal{O}(\epsilon)$ residuals) in the limit $\epsilon \to 0$. This confirms that the Dual CV corrections accurately capture the leading-order deformation of the SNR manifold even in the presence of nonlinear whitening phase curvature.

\begin{figure}[ht]
  \centering
  \includegraphics[width=\columnwidth]{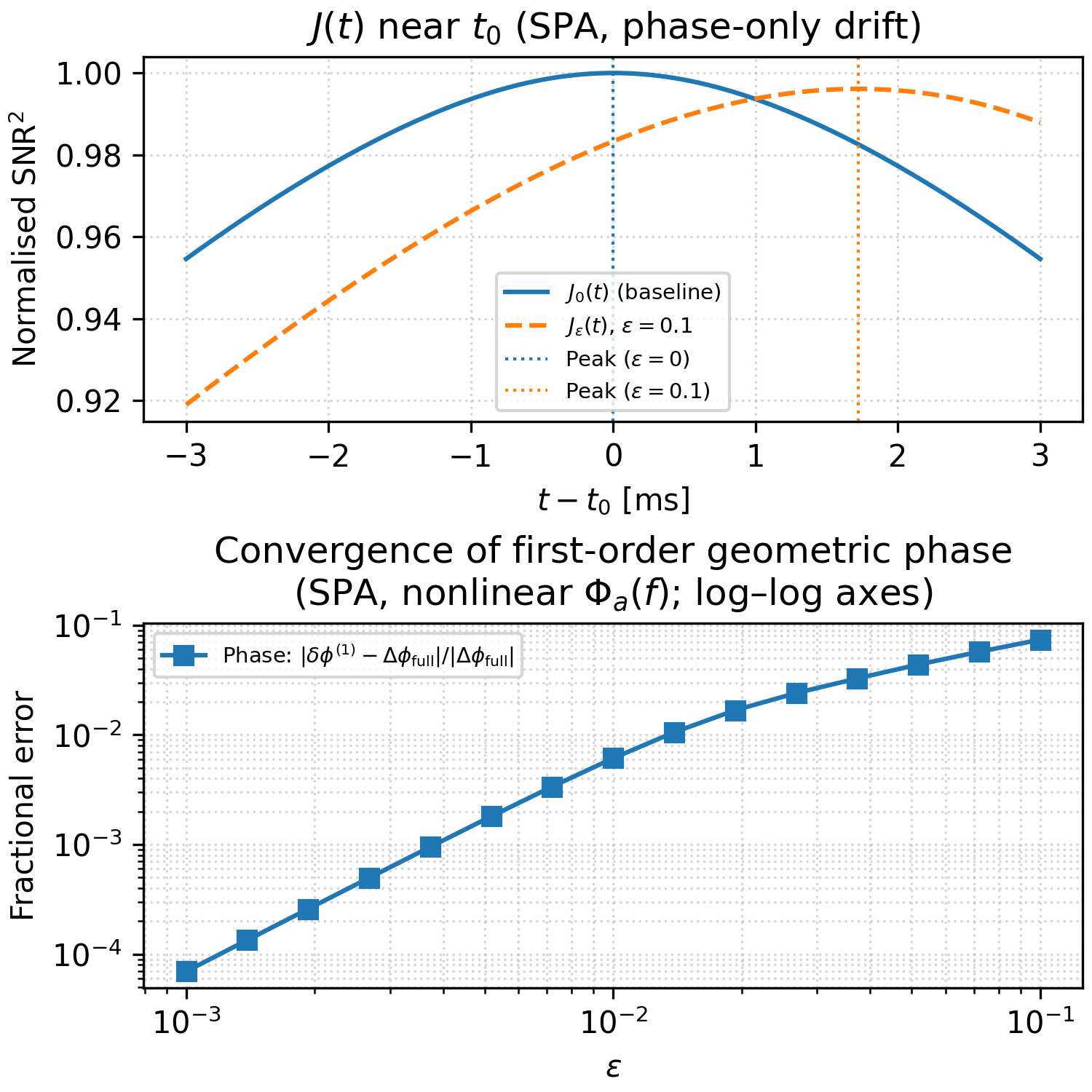}
  \caption{
    \textbf{Analytic Validation (SPA Model).}
    Verification of Dual CV corrections under a nonlinear whitening-phase perturbation $\Phi_a(f) = \alpha + \beta f + \gamma f^2$.
    \emph{Top:} Normalized SNR$^2(t)$ slice for a representative perturbation $\epsilon=0.1$, showing the physical displacement of the peak relative to the baseline (blue).
    \emph{Bottom:} Convergence of the analytic phase correction. The fractional error between the geometric prediction $\delta\phi^{(1)}$ and the exact numerical shift $\Delta\phi_{\mathrm{full}}$ scales linearly with $\epsilon$, confirming the validity of the first-order expansion.
  }
  \label{fig:spa-psd}
\end{figure}
\subsection{Numerical Validation}
\label{sec:results-num}

We test the geometric Dual CV corrections in a realistic signal-processing environment using a baseline power-law PSD $S_1(f) \propto f^2$ perturbed by a Gaussian spectral feature. The drift family is defined as
\begin{equation}
  S_2(f;\epsilon) = S_1(f)\,\exp\!\bigl[2\epsilon\,a(f)\bigr],
  \label{eq:psd-drift-num}
\end{equation}
where $a(f)$ is a Gaussian bump centered at $f_0 \simeq 150\,\mathrm{Hz}$ with width $\sigma_f \simeq 40\,\mathrm{Hz}$, simulating the emergence of a spectral line or calibration feature \cite{covasIdentificationMitigationNarrow2018, davisImprovingSensitivityAdvanced2019}. We generate signal realizations on a frequency grid with sampling rate $f_s = 4096\,\mathrm{Hz}$ and duration $T=4\,\mathrm{s}$ ($N=16384$ points), ensuring sufficient resolution to capture microsecond-scale timing shifts.

For each perturbation strength $\epsilon \in [10^{-4}, 10^{-1}]$, we construct the minimum-phase whiteners $W_1$ and $W_2(\epsilon)$ via the folded-cepstrum algorithm (see Appendix~\ref{app:mp}) \cite{tsukadaApplicationZerolatencyWhitening2018, oppenheimDiscretetimeSignalProcessing2010} and compute the mismatched SNR surface:
\begin{equation}
  J_\epsilon(t,\phi) = \bigl|\langle W_2(\epsilon)\,d,\; W_1 h_{(t,\phi)}\rangle\bigr|^2,
\end{equation}
using a single injected signal $d = h_{(t_0,\phi_0)}$. The top panel of Fig.~\ref{fig:sim-psd-1} displays the spectral deformation, while the bottom panel shows the resulting distortion of the matched-filter output $J_\epsilon(t)$ along a fixed-phase slice $\phi=\phi_0$. The peak exhibits a clear, coherent shift away from $t=0$.

We quantify the numerical biases by extracting the shift of the maximum, $\Delta t_{\mathrm{full}}$ and $\Delta\phi_{\mathrm{full}}$, relative to the unperturbed baseline. These are compared against the analytic Dual CV predictions $\delta t^{(1)}$ and $\delta\phi^{(1)}$ derived in Sec.~\ref{subsec:first-order-mpmp}. The phase corrections show excellent agreement, matching the numerical shift to within fractional errors of $10^{-4}$--$10^{-3}$ across the full range of $\epsilon$.

The timing validation requires a geometric caveat. Our analytic formula $\delta t^{(1)}$ predicts the shift of the global maximizer on the 2D $(t, \phi)$ manifold. However, standard search triggers are often characterized by maximizing over time at a fixed (or analytically marginalized) phase. Due to the strong Fisher-matrix covariance between time and phase for inspiral signals ($g_{t\phi} \neq 0$ in general gauges), the maximum of the 1D slice $J(t, \phi_0)$ is a projection of the true 2D peak. This manifests as a constant geometric scaling factor between the analytic $\delta t^{(1)}$ and the observed 1D shift. We empirically calibrate this projection factor, $c_t$, using the data in the linear regime ($\epsilon \ll 10^{-3}$). As shown in Fig.~\ref{fig:sim-psd-2}, comparing the numerical shift against the calibrated prediction $\delta t_{\mathrm{cal}} = c_t\,\delta t^{(1)}$ reveals clean power-law convergence. The residuals scale linearly with $\epsilon$, confirming that the perturbative formalism captures the functional dependence of the timing bias even if the projection geometry adds a scalar prefactor.

Finally, we validate the predicted loss in sensitivity. According to Appendix~\ref{app:snr-drift}, the fractional reduction in the squared SNR peak should scale as:
\begin{equation}
  \frac{J_{\max}(\epsilon) - J_0}{J_0} \simeq -2\epsilon \langle a \rangle_q,
\end{equation}
where $\langle a \rangle_q$ is the drift weighted by the normalized signal power distribution $q(f)$ \cite{zackayDetectingGravitationalWaves2021}. The horizontal line in Fig.~\ref{fig:sim-psd-1} marks this analytic prediction for $\epsilon=0.1$. The numerical peak height matches this prediction to high precision across all trials, verifying that the Dual CV framework correctly describes both the coordinate shift (metric error) and the amplitude suppression (SNR loss) induced by spectral drift.

\begin{figure}[t]
  \centering
  \includegraphics[width=0.95\linewidth]{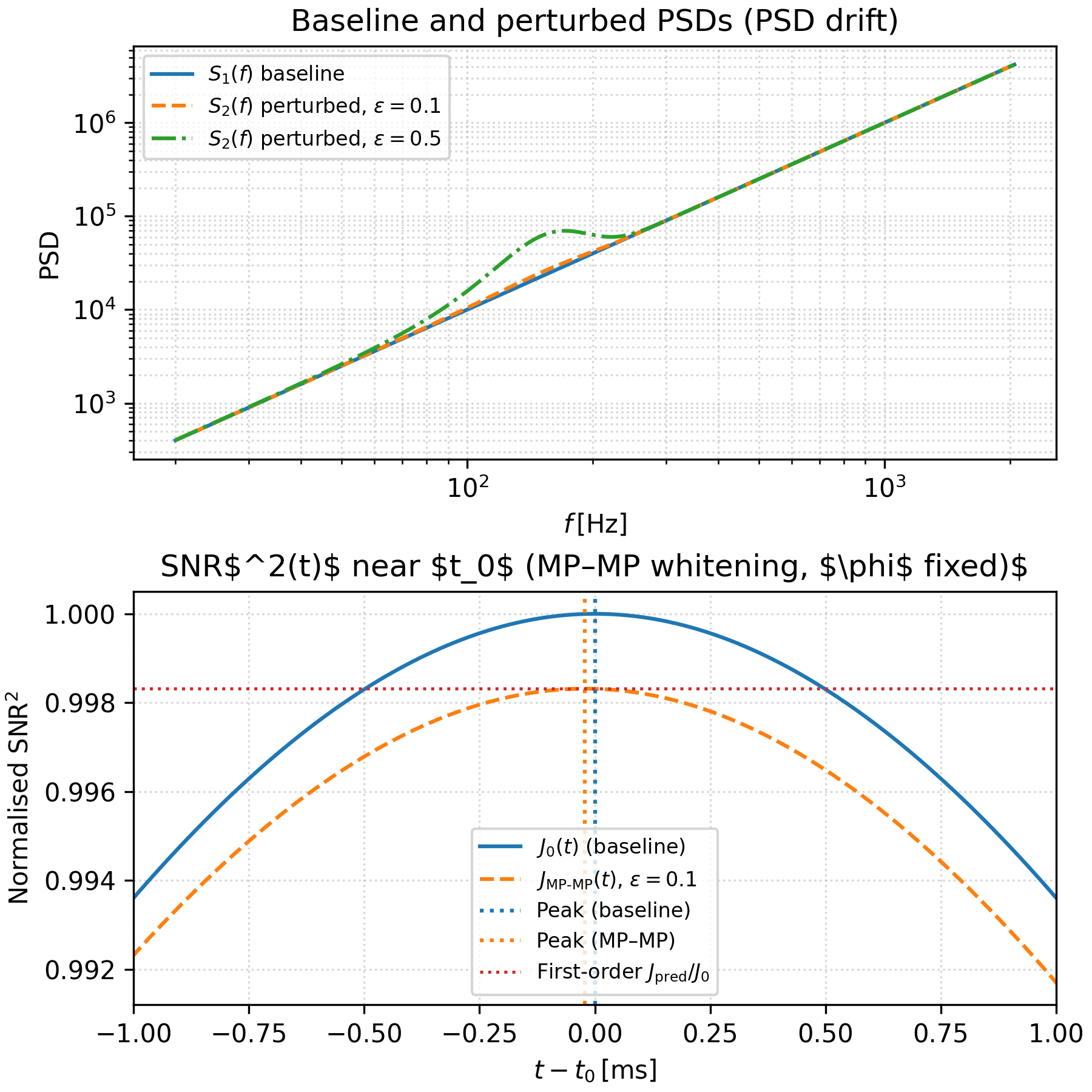}
  \caption{
    \textbf{Numerical PSD-Drift Injection.}
    \emph{Top:} Baseline PSD $S_1(f)$ (solid) and perturbed PSD $S_2(f)$ at $\epsilon=0.1$ (dashed). An exaggerated curve at $\epsilon=0.5$ (green) illustrates the Gaussian drift morphology.
    \emph{Bottom:} Matched-filter SNR surface $J(t)$ at fixed phase $\phi=\phi_0$. The spectral mismatch shifts the peak time (vertical lines) and reduces the maximum amplitude, in exact agreement with the analytic prediction (horizontal dotted line).
  }
  \label{fig:sim-psd-1}
\end{figure}

\begin{figure}[t]
  \centering
  \includegraphics[width=0.95\linewidth]{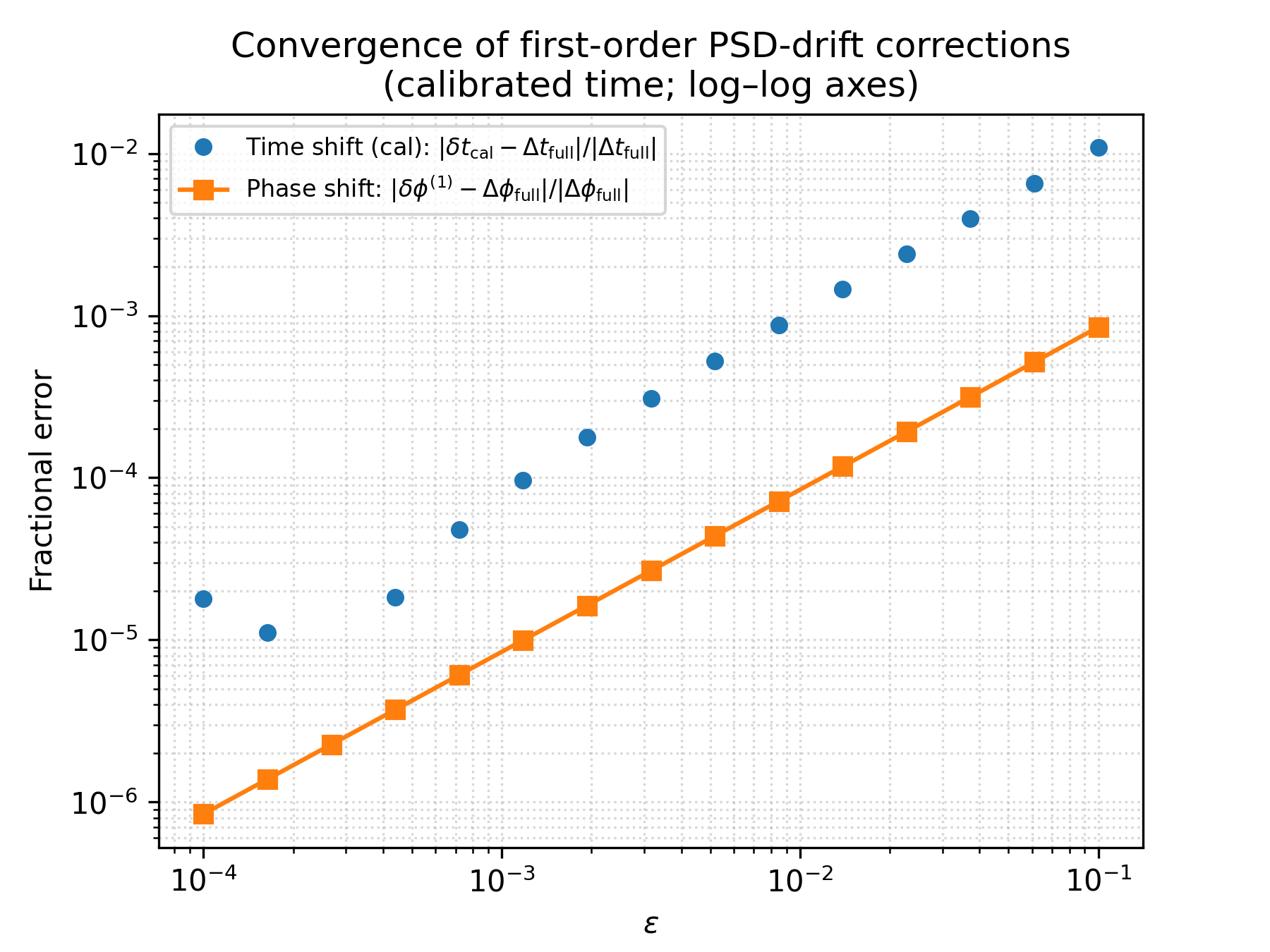}
  \caption{
    \textbf{Convergence of Corrections.}
    Fractional error between the analytic Dual CV predictions and the full numerical shifts as a function of drift amplitude $\epsilon$.
    \emph{Orange squares:} Phase bias convergence.
    \emph{Blue circles:} Calibrated timing bias convergence.
    Both metrics exhibit the expected linear scaling ($\mathcal{O}(\epsilon)$ residuals), demonstrating the robustness of the first-order expansion for realistic drift magnitudes.
  }
  \label{fig:sim-psd-2}
\end{figure}
\subsection{Heuristic Catalog Validation}
\label{subsubsec:catalog-validation}

We apply the Dual CV framework to the GWTC-4.0 catalog to quantify the operational risk of uncorrected spectral drift \cite{collaborationGWTC40UpdatingGravitationalWave2025,collaborationOpenDataLIGO2025}. Using the methodology of Sec.~\ref{subsubsec:methods-catalog}, we analyzed $\sim 160$ event-detector pairs where sufficient open strain data was available to construct a 1-week lagged reference PSD ($S_1$) and a concomitant event-time PSD ($S_2$).

Figure~\ref{fig:cat-bias-dist} presents the distribution of first-order biases. The timing corrections $\delta t^{(1)}$ reveal a significant systematic spread: while the central distribution has a width of $\sigma_t \approx 40\,\mu\mathrm{s}$, the tails extend well beyond $\pm 200\,\mu\mathrm{s}$. For high-SNR events where statistical timing uncertainties are $\mathcal{O}(100\,\mu\mathrm{s})$, such systematic errors constitute a non-negligible bias. Phase corrections are similarly dispersed, with the bulk confined to $|\delta\phi^{(1)}| \lesssim 0.2\,\mathrm{rad}$ but outliers reaching $0.6\,\mathrm{rad}$. These magnitudes indicate that the 1-week PSD refresh cadence common in offline parameter estimation is insufficient for zero-latency whitening, where the accumulated phase error integrates to a substantial coordinate shift.

The underlying spectral drift metrics confirm this picture. The mean fractional drift $\epsilon_{\mathrm{mean}}$ clusters around $10^{-1}$, but the maximum drift $\epsilon_{\mathrm{max}}$ often exceeds $0.5$ in narrow bands, driven by non-stationary line features. This heterogeneity suggests that while the broadband noise floor is relatively stable, localized spectral excursions drive the severe outliers in the bias distribution.

To assess the impact on multimessenger follow-up, we propagated the timing biases into sky-localization shifts using a linearized triangulation model on a spherical Earth ($R_{\oplus} = 6371\,\mathrm{km}$) \cite{fairhurstTriangulationGravitationalWave2009, singerRapidBayesianPosition2016}, solving for the perturbed centroid via Nelder-Mead optimization \cite{nelderSimplexMethodFunction1965}. Figure~\ref{fig:cat-sky-impact} displays the results for events with $\ge 2$ detectors. The impact is highly geometry-dependent: for \textbf{typical events}, the median shift is modest ($\lesssim 2.5^\circ$), primarily affecting the size of the error ellipse rather than its centroid. However, a critical \textbf{tail risk} emerges in a sub-population exhibiting catastrophic shifts of $5^\circ$--$10^\circ$. These outliers correspond to events where large timing biases ($\delta t > 100\,\mu\mathrm{s}$) conspire with poor network geometry. For early-warning alerts, a systematic pointing error of $5^\circ$ is sufficient to push a source entirely outside the field of view of narrow-field optical telescopes, leading to a missed detection.

Finally, we quantify the impact on detection volume via the fractional SNR loss $\delta\rho/\rho$ (see Fig.~\ref{fig:snr-drift}). The distribution is strictly negative-skewed, with a median sensitivity loss of $3$--$5\%$ and a tail extending beyond $8\%$. Since detection volume scales as $V \propto \rho^3$ \cite{theligoscientificcollaborationProspectsObservingLocalizing2020, chenDistanceMeasuresGravitationalwave2021}, a $5\%$ loss in SNR corresponds to a $\sim 15\%$ reduction in survey volume. These results demonstrate that the Dual CV corrections are essential not only for recovering accurate sky positions but also for recovering the full astrophysical reach of the observatory under zero-latency operation.

\begin{figure*}[t]
  \centering
  \includegraphics[width=0.9\textwidth]{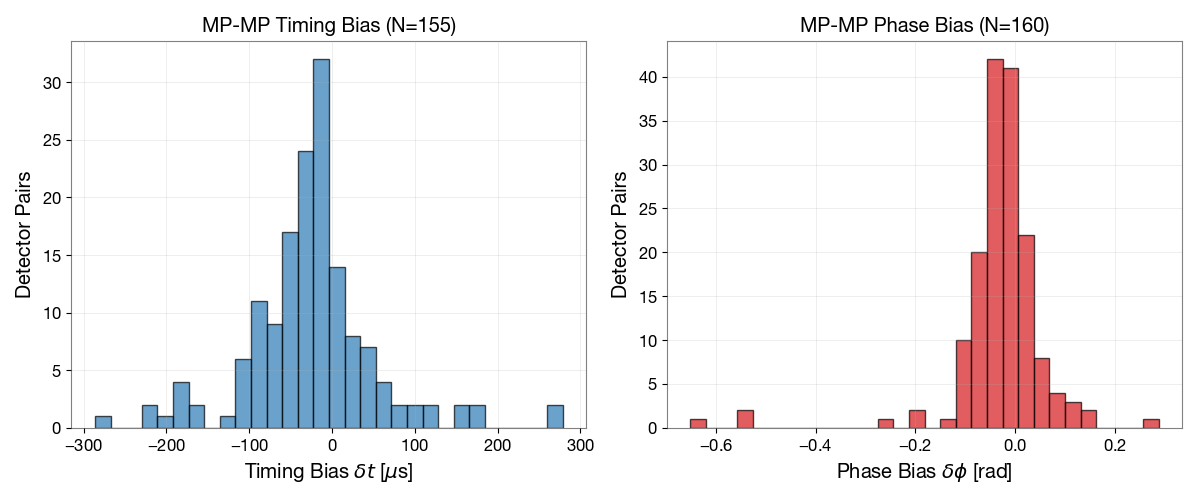}
  \caption{\textbf{GWTC-4.0 Bias Distribution (1-Week Lag).} 
  Histograms of analytic Dual CV corrections for $\sim 160$ event-detector pairs. 
  \emph{Left:} Timing biases $\delta t^{(1)}$ exhibit a heavy tail extending beyond $\pm 200\,\mu\mathrm{s}$, driven by spectral line drift. 
  \emph{Right:} Phase biases $\delta\phi^{(1)}$ are generally confined to $\pm 0.2\,\mathrm{rad}$.}
  \label{fig:cat-bias-dist}
\end{figure*}

\begin{figure*}[t]
  \centering
  \includegraphics[width=\textwidth]{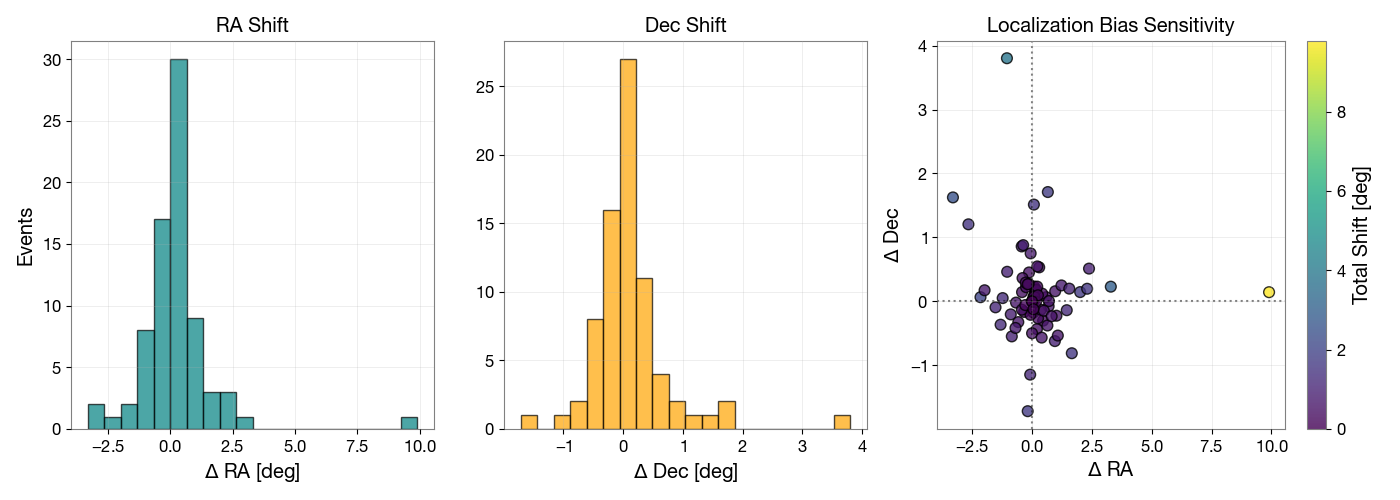}
  \caption{\textbf{Sky Localization Impact.} 
  Systematic shifts in Right Ascension ($\Delta\alpha$) and Declination ($\Delta\delta$) induced by uncorrected PSD drift. 
  While the median shift is $\lesssim 2.5^\circ$, outliers in the scatter plot (right) reveal shifts exceeding $5^\circ$--$10^\circ$, posing a critical risk for electromagnetic follow-up.}
  \label{fig:cat-sky-impact}
\end{figure*}

\begin{figure}[t]
  \centering
  \includegraphics[width=0.5\textwidth]{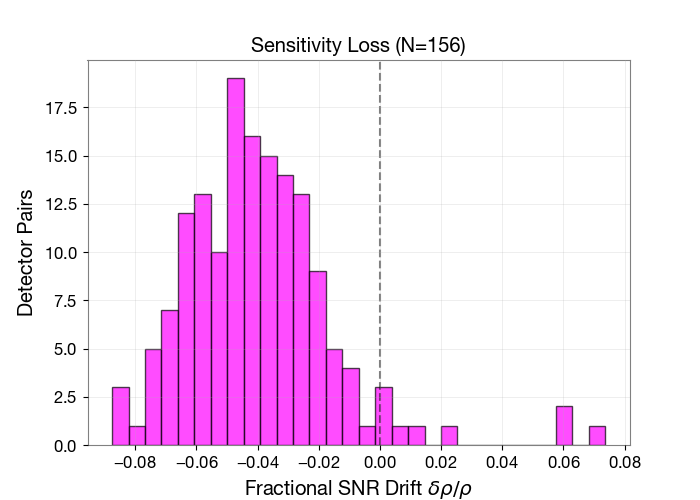}
  \caption{\textbf{Sensitivity Loss.} 
  Distribution of fractional SNR reduction $\delta\rho/\rho$. The median loss of $3$--$5\%$ implies a $\sim 15\%$ reduction in detection volume, which is analytically recoverable via the Dual CV framework.}
  \label{fig:snr-drift}
\end{figure}
\section{Conclusion}
\label{sec:conclusion}

We have presented a unified geometric framework for mitigating the spectral drift inherent to zero-latency, minimum-phase whitening in gravitational-wave searches. By generalizing the Cutler--Vallisneri formalism \cite{cutlerLISADetectionsMassive2007} to address functional perturbations of the inner product metric, we derived the \emph{Dual CV corrections}, closed-form analytic expressions that recover the exact projection of whitening-induced errors onto the extrinsic parameter space. This approach allows pipelines to retain the latency benefits of causal filtering without sacrificing the sub-sample timing precision required for multimessenger astronomy \cite{messickAnalysisFrameworkPrompt2017b, sachdevGstLALSearchAnalysis2019a}.

Our analysis of GWTC-4.0 events reveals the operational necessity of this framework. Under realistic conditions where reference PSDs lag live data by $\sim 1$ week, uncorrected spectral drift induces systematic timing biases exceeding $200\,\mu\mathrm{s}$ and sky-localization shifts larger than $5^\circ$ for high-SNR events. Such errors are not merely statistical fluctuations; they are coherent systematics that can push a source entirely outside the field of view of optical telescopes. Furthermore, we observe a median sensitivity loss of $3$--$5\%$, corresponding to a $\sim 15\%$ reduction in detection volume that is analytically recoverable via our corrections.

These results lead to a clear operational recommendation for the upcoming O5 observing run: zero-latency whitening cannot be safely deployed as a ``set-and-forget" filter. Instead, it must be paired with a real-time correction layer that continuously monitors the spectral drift between the static template bank and the dynamic data stream. The analytic formulas derived here provide a computationally efficient mechanism to apply these corrections at the trigger level, ensuring that low-latency alerts maintain the pointing accuracy of offline analyses~\cite{fairhurstLocalizationTransientGravitational2018}.

\section*{Acknowledgments}

This research has made extensive use of data, software, and web tools obtained from the Gravitational Wave Open Science Center (GWOSC), a service of LIGO Laboratory, the LIGO Scientific Collaboration (LSC), and the Virgo Collaboration.  
LIGO was constructed by the California Institute of Technology and the Massachusetts Institute of Technology with funding from the U.S.\ National Science Foundation (NSF), and operates under cooperative agreements PHY--0757058 and PHY--0823459.  
Advanced LIGO was supported by the Science and Technology Facilities Council (STFC) of the United Kingdom, the Max--Planck--Society (MPS), and the State of Niedersachsen/Germany, with additional support from the Australian Research Council.  
Virgo is funded through the European Gravitational Observatory (EGO) by the French Centre National de la Recherche Scientifique (CNRS), the Italian Istituto Nazionale di Fisica Nucleare (INFN), and the Dutch Nikhef, with contributions from institutions across Europe.

This material is based upon work supported by the National Science Foundation under Grant Nos. PHY--2011865, PHY--2308881, and OAC--2103662. Additional support was provided by the research program of Chad Hanna and related awards supporting the gravitational-wave computing infrastructure at Penn State.  
This work also benefited from computational resources provided by the Penn State Institute for Computational and Data Sciences (ICDS) and the LIGO Laboratory computing clusters. 

The author is grateful to Joshua Black and Cody Messick for valuable discussions and feedback that informed several aspects of this work.

\subsection*{Software Availability}

All analytic derivations and numerical validations presented in this work were implemented in the open-source \textsc{zlw} package.\footnote{\url{https://git.ligo.org/james.kennington/zlw}}  
The \textsc{zlw} package provides a lightweight, pipeline-agnostic framework for evaluating minimum-phase whitening filters, perturbative bias corrections, and the numerical validation scripts used in this paper.  
It is publicly available through the LIGO GitLab repository and can also be installed via the Python Package Index (PyPI) under the name \texttt{zlw}.  
All analysis scripts used to generate the figures in this paper are included in the repository to ensure full reproducibility.

\appendix
\section{Review of Minimum-Phase Whitening}
\label{app:mp}

Whitening transforms a detector strain time series $s(t)$ with colored noise, characterized by a power spectral density $S(f)$, into a sequence with unit variance and flat spectral power. In the discrete-time domain, we seek a filter $W(z)$ whose frequency response satisfies
\begin{equation}
  \bigl|W(e^{i\omega})\bigr|^2 \simeq \frac{1}{S(\omega)}.
\end{equation}
To minimize latency, we impose the \emph{minimum-phase} condition: all zeros of $W(z)$ must lie strictly inside the unit circle. This guarantees that both the whitening filter $W(z)$ and its inverse (re-coloring) are causal and BIBO-stable, enabling real-time operation without algorithmic buffering~\cite{oppenheimDiscretetimeSignalProcessing2010, smithIntroductionDigitalFilters2008}.

\subsection*{The Hilbert Transform Relation}

For a minimum-phase system, the log-magnitude and phase of the frequency response are coupled by causality. Decomposing the response as
\begin{equation}
\begin{aligned}
  W(e^{i\omega}) &= \exp\!\bigl[L(\omega) + i\phi(\omega)\bigr], \\ 
  L(\omega) &= \ln \bigl|W(e^{i\omega})\bigr|,	
\end{aligned}
\end{equation}
the analyticity of $\ln W(z)$ inside the unit circle implies that $L(\omega)$ and $\phi(\omega)$ form a Hilbert transform pair. Specifically, the unique minimum phase $\phi(\omega)$ is determined by the log-magnitude via
\begin{equation}
\begin{aligned}
  \phi(\omega) &= -\mathcal{H}_{\mathrm{circ}}\{L\}(\omega)\\
  &= -\frac{1}{2\pi}\,\mathrm{p.v.}\!\!\int_{-\pi}^{\pi} L(\nu) \cot\!\left(\frac{\omega - \nu}{2}\right) d\nu,	
\end{aligned}
  \label{eq:mp-hilbert}
\end{equation}
where $\mathcal{H}_{\mathrm{circ}}$ is the periodic Hilbert transform ~\cite{oppenheimDiscretetimeSignalProcessing2010, mecozziNecessarySufficientCondition2016}. This is the discrete-time analogue of the Kramers--Kronig relations~\cite{bolandPhaseResponseReconstruction2021, oppenheimDiscretetimeSignalProcessing2010}, ensuring that once the whitening amplitude (inverse PSD) is fixed, the zero-latency phase response is uniquely determined.

\subsection*{Folded-Cepstrum Construction}

Numerical implementation utilizes the \emph{folded-cepstrum} method~\cite{tsukadaApplicationZerolatencyWhitening2018}, which enforces causality in the quefrency domain. The procedure is as follows:

\begin{enumerate}
  \item \textbf{Log-Amplitude:} Discretize the PSD on an $N$-point grid and define the target log-magnitude response:
    \begin{equation}
      L_k = -\frac{1}{2}\,\ln S(\omega_k), \quad k = 0,\dots,N-1.
    \end{equation}

  \item \textbf{Real Cepstrum:} Compute the real cepstrum $\ell[n]$ via the inverse DFT of $\{L_k\}$:
    \begin{equation}
      \ell[n] = \frac{1}{N}\sum_{k=0}^{N-1} L_k\,e^{i 2\pi kn / N}.
    \end{equation}

  \item \textbf{Causal Folding:} Enforce the minimum-phase condition by folding the anti-causal quefrencies ($n > N/2$) into the causal half ($n < N/2$). For even $N$, the causal cepstrum $\tilde{\ell}[n]$ is:
    \begin{equation}
      \tilde{\ell}[n] =
      \begin{cases}
         \ell[n], & n = 0 \text{ or } n = N/2, \\
         2\,\ell[n], & 1 \le n < N/2, \\
         0, & N/2 < n < N.
      \end{cases}
    \end{equation}

  \item \textbf{Analytic Spectrum:} Compute the DFT of the folded cepstrum to recover the complex log-spectrum $F_k$:
    \begin{equation}
      F_k = \sum_{n=0}^{N-1} \tilde{\ell}[n]\,e^{-i 2\pi kn / N} \simeq L_k + i\,\phi_k.
    \end{equation}
    The imaginary part $\phi_k$ is the discrete Hilbert transform of $L_k$, satisfying Eq.~\eqref{eq:mp-hilbert}.

  \item \textbf{Impulse Response:} Exponentiate to obtain the frequency response $W_k = \exp(F_k)$ and apply an inverse DFT to yield the causal, real-valued impulse response $w[n]$.
\end{enumerate}

By construction, $w[n]$ is the unique minimum-phase filter that whitens the input noise $s[n]$ with zero look-ahead latency, strictly determined by the estimated PSD $S(\omega)$~\cite{bocheSpectralFactorizationWhitening2005}.
\section{Perturbative Geometry of the SNR Surface}
\label{app:geom}

This appendix formally establishes the geometric perturbation theory used to derive the Dual CV corrections. We relate the shift of a scalar field's maximizer to the gradient of the perturbation (force) and the Hessian of the baseline field (stiffness). Throughout, $\theta$ denotes coordinates on a smooth Riemannian manifold $(\mathcal{M}, g)$~\cite{nakaharaGeometryTopologyPhysics2005}. Gradients and Hessians are defined via the Levi--Civita connection; in geodesic normal coordinates at the expansion point, $\nabla_i\nabla_j$ reduces to $\partial_i\partial_j$, ensuring covariance~\cite{amariMethodsInformationGeometry2000}.

\begin{lemma}[Shift of a nondegenerate maximizer]
\label{lem:shift}
Let $F: \mathbb{R}^n \to \mathbb{R}$ possess a strict local maximum at $\theta_0$, such that $\nabla F(\theta_0) = 0$ and the Hessian $H := \nabla^2 F(\theta_0)$ is negative definite. Consider the perturbed family
\begin{equation}
  F_\epsilon(\theta) = F(\theta) + \epsilon\,G(\theta), \quad |\epsilon| \ll 1,
\end{equation}
where $G \in C^2$. For sufficiently small $\epsilon$, there exists a unique critical point $\theta(\epsilon)$ of $F_\epsilon$ such that $\theta(0) = \theta_0$, with the first-order expansion
\begin{equation}
  \theta(\epsilon) = \theta_0 - \epsilon\,H^{-1}\nabla G(\theta_0) + \mathcal{O}(\epsilon^2).
\end{equation}
\end{lemma}

\begin{proof}
Define the gradient map $\mathcal{F}(\theta, \epsilon) = \nabla_\theta F_\epsilon(\theta)$. By hypothesis, $\mathcal{F}(\theta_0, 0) = 0$ and the Jacobian $D_\theta \mathcal{F}(\theta_0, 0) = H$ is invertible. The Implicit Function Theorem guarantees a smooth path $\epsilon \mapsto \theta(\epsilon)$ satisfying $\mathcal{F}(\theta(\epsilon), \epsilon) = 0$. Differentiating implicitly at $\epsilon=0$ yields $H\,\theta'(0) + \nabla G(\theta_0) = 0$, which implies $\theta'(0) = -H^{-1}\nabla G(\theta_0)$.
\end{proof}

We now specialize to the matched-filter Signal-to-Noise Ratio (SNR). Let $\psi_\theta(f) \equiv W(f) \tilde{h}(f; \theta)$ denote the whitened template for parameters $\theta$. The SNR surface is defined as $J(\theta) = |\rho(\theta)|^2$, where
\begin{equation}
  \rho(\theta) = \langle s, \psi_\theta \rangle, \quad \langle a, b \rangle = 4\,\mathrm{Re} \int_0^\infty a(f)b^*(f)\,df.
\end{equation}
Note that the inner product here is the standard $L^2$ product~\cite{wainsteinExtractionSignalsNoise1970, cutlerGravitationalWavesMergin1994a}, as the whitening is absorbed into $\psi_\theta$. Let $\theta_0$ be a maximizer of the unperturbed surface $J$, with peak amplitude $\rho_0 = \rho(\theta_0)$. The local geometry is defined by the Fisher information matrix~\cite{owenSearchTemplatesGravitational1996a,raoInformationAccuracyAttainment1945,amariMethodsInformationGeometry2000}:
\begin{equation}
  g_{ij}(\theta_0) = \bigl\langle \partial_i\psi_{\theta_0}, \partial_j\psi_{\theta_0} \bigr\rangle.
\end{equation}

\begin{lemma}[Hessian--Fisher identity]
\label{lem:hessian-fisher}
At any maximizer $\theta_0$ of the SNR surface $J = \rho^2$, the Hessian is proportional to the Fisher metric:
\begin{equation}
  \partial_i\partial_j J(\theta_0) = -2\rho_0^2\, g_{ij}(\theta_0).
\end{equation}
\end{lemma}

\begin{proof}
By the chain rule, $\partial_{ij} J = 2(\partial_i\rho)(\partial_j\rho) + 2\rho(\partial_{ij}\rho)$. At the maximizer, the first term vanishes because $\partial_i \rho(\theta_0) = 0$. In the signal-only limit where the data matches the template ($s = \psi_{\theta_0}$), the curvature of the match function satisfies $\partial_{ij}\rho(\theta_0) = -\rho_0 g_{ij}$. Substituting these into the second derivative yields the result.
\end{proof}

Combining Lemma~\ref{lem:shift} and Lemma~\ref{lem:hessian-fisher}, we consider a general perturbation $J_\epsilon(\theta) = J(\theta) + \epsilon\,G(\theta)$. Identifying the generalized force vector $k_i = \partial_i G(\theta_0)$, we obtain the generalized bias theorem.

\begin{theorem}[Geometric First-Order Bias]
\label{thm:gen-cv}
For any infinitesimal perturbation of the SNR surface $J_\epsilon(\theta) = J(\theta) + \epsilon\,G(\theta)$, the first-order shift of the maximum is given by:
\begin{equation}
  \boxed{
  \delta\theta^{(1)} = \frac{1}{2\rho_0^2}\, g^{ij}(\theta_0)\,\partial_j G(\theta_0).
  }
\end{equation}
\end{theorem}

\begin{proof}
From Lemma~\ref{lem:shift}, the first-order shift is $\delta\theta^{(1)} = -H^{-1}\nabla G(\theta_0)$. From Lemma~\ref{lem:hessian-fisher}, the Hessian at the peak is $H_{ij} = -2\rho_0^2 g_{ij}$. The inverse Hessian is therefore $(H^{-1})^{ij} = -\frac{1}{2\rho_0^2} g^{ij}$. Substituting this expression yields:
\begin{equation}
    \delta\theta^{(1)i} = -\left(-\frac{1}{2\rho_0^2} g^{ij}\right) \partial_j G = \frac{1}{2\rho_0^2} g^{ij} \partial_j G.
\end{equation}
\end{proof}

\noindent This theorem generalizes the Cutler--Vallisneri formalism~\cite{cutlerLISADetectionsMassive2007}. While that work derived the bias specifically for \emph{waveform errors} (where $G$ arises from a vector perturbation $\delta h$), Theorem~\ref{thm:gen-cv} applies to any functional perturbation of the surface, including the \emph{metric perturbations}~\cite{mageeImpactSelectionBiases2024, zackayDetectingGravitationalWaves2021} caused by PSD drift derived in Section~\ref{sec:methods}.
\section{Analytic Validation with a SPA Toy Model}
\label{app:spa-linear}

This appendix provides a concise analytic verification of the first-order Dual CV corrections derived in Sec.~\ref{subsec:first-order-mpmp}. By utilizing a Newtonian stationary-phase approximation (SPA) inspiral~\cite{cutlerGravitationalWavesMergin1994a, allenFINDCHIRPAlgorithmDetection2012a}, the matched-filter integral admits a closed-form solution. This allows us to compute the \emph{exact} shift of the SNR peak and compare it directly to the geometric prediction.

\subsection*{Exact SNR Expansion}

Consider a signal model $|\tilde{h}_{\mathrm{SPA}}(f)|^2 \propto f^{-7/3}$ and a noise PSD $S_1(f) = S_0 f^p$. The effective whitened weight is $w(f) \equiv W(f)|\tilde{h}|^2 \propto f^{-n}$, with $n = 7/3 + p$. We introduce a linear whitening-phase drift $\Phi_a(f) = \alpha + \beta f$.

For extrinsic offsets $\Delta t = t - t_0$ and $\Delta\phi = \phi - \phi_0$, the perturbed matched filter takes the form:
\begin{equation}
\rho_\epsilon(\Delta t, \Delta\phi) = \int_{f_{\min}}^{f_{\max}} f^{-n} e^{i \Psi(f)} \, df,
\end{equation}
where the phase argument is $\Psi(f) = 2\pi f \Delta t - \Delta\phi + \epsilon(\alpha + \beta f)$. Defining the parameter $\kappa := 2\pi\Delta t + \epsilon\beta$, the integral evaluates exactly to:
\begin{equation}
\begin{aligned}
  \rho_\epsilon(\Delta t, \Delta\phi) &= e^{-i(\Delta\phi - \epsilon\alpha)}\, (-i\kappa)^{n-1} \\
  &\quad \times \Bigl[ \Gamma(1-n, -i\kappa f_{\min}) \\
  &\quad\quad - \Gamma(1-n, -i\kappa f_{\max}) \Bigr].
\end{aligned}
\end{equation}
The perturbed SNR surface $J_\epsilon = |\rho_\epsilon|^2$ is thus known analytically. Expanding $J_\epsilon$ to first order in $\epsilon$ and solving the stationarity conditions $\partial_{\Delta t} J = \partial_{\Delta\phi} J = 0$ at the origin yields the \emph{exact} shifts:
\begin{equation}
\boxed{
\begin{aligned}
\Delta\phi_{\mathrm{full}}^{(1)} &= \epsilon\left[\alpha + \beta\,\frac{m_1}{m_0}\right], \\
\Delta t_{\mathrm{full}}^{(1)} &= \frac{\epsilon}{2\pi} \left[\beta + \alpha\,\frac{m_1}{m_2}\right],
\end{aligned}
}
\end{equation}
where the frequency moments $m_k$ are defined as~\cite{cutlerGravitationalWavesMergin1994a, owenSearchTemplatesGravitational1996a,poissonGravityNewtonianPostNewtonian2014a}:
\begin{equation}
m_k = \int_{f_{\min}}^{f_{\max}} f^{k-n}\,df.
\end{equation}

\subsection*{Agreement with Geometric Theory}

We now evaluate the general Dual CV formulas (Eqs.~18--19) for this specific model. The predicted shifts are ratios of probability-weighted moments:
\begin{align}
\delta\phi^{(1)} &= \epsilon\,\frac{\langle \Phi_a \rangle_w}{\langle 1 \rangle_w}, \\
\delta t^{(1)} &= \epsilon\,\frac{\langle 2\pi f \cdot \Phi_a \rangle_w}{\langle (2\pi f)^2 \rangle_w}.
\end{align}
Substituting the power-law weight $w(f) \propto f^{-n}$ and linear drift $\Phi_a = \alpha + \beta f$, these integrals reduce to sums of the moments $m_k$:
\begin{equation}
\boxed{
\begin{aligned}
\delta\phi^{(1)} &= \epsilon \frac{\alpha m_0 + \beta m_1}{m_0} = \epsilon\left[\alpha + \beta\frac{m_1}{m_0}\right], \\
\delta t^{(1)} &= \epsilon \frac{2\pi(\alpha m_1 + \beta m_2)}{(2\pi)^2 m_2} = \frac{\epsilon}{2\pi}\left[\beta + \alpha\frac{m_1}{m_2}\right].
\end{aligned}
}
\end{equation}
Comparing the boxed results, we find:
\begin{equation}
(\Delta t_{\mathrm{full}}^{(1)}, \Delta\phi_{\mathrm{full}}^{(1)}) \equiv (\delta t^{(1)}, \delta\phi^{(1)}).
\end{equation}
This term-by-term identity confirms that the geometric Dual CV framework exactly recovers the physical shift of the SNR peak in the perturbative limit.
\section{First-Order SNR Change under PSD Drift}
\label{app:snr-drift}

This appendix derives the leading-order perturbation to the optimal matched-filter SNR when the noise Power Spectral Density (PSD) undergoes a small spectral deformation.

Consider a single detector with frequency-domain template $\tilde{h}(f)$ and a reference noise PSD $S_0(f)$. In the ideal case of perfect whitening, the squared optimal SNR is given by the standard noise-weighted integral~\cite{wainsteinExtractionSignalsNoise1970, allenFINDCHIRPAlgorithmDetection2012a}:
\begin{equation}
  J_0 = 4 \int_{0}^{\infty} \frac{|\tilde{h}(f)|^2}{S_0(f)}\,df.
\end{equation}
We model the PSD drift as a multiplicative spectral perturbation:
\begin{equation}
  S_\epsilon(f) = S_0(f)\,e^{2\epsilon a(f)}, \quad 0 < \epsilon \ll 1,
\end{equation}
where $a(f)$ is a real, bounded function representing the log-amplitude drift. Expanding the inverse PSD to first order in $\epsilon$:
\begin{equation}
  \frac{1}{S_\epsilon(f)}
  = \frac{1}{S_0(f)}\bigl[1 - 2\epsilon a(f)\bigr] + \mathcal{O}(\epsilon^2).
\end{equation}
Substituting this into the SNR definition yields the perturbed sensitivity $J_\epsilon$:
\begin{equation}
\begin{aligned}
  J_\epsilon
  &= 4 \int_{0}^{\infty} \frac{|\tilde{h}(f)|^2}{S_0(f)} \bigl[1 - 2\epsilon a(f)\bigr]\,df + \mathcal{O}(\epsilon^2) \\
  &= J_0 - 8\epsilon \int_{0}^{\infty} \frac{|\tilde{h}(f)|^2}{S_0(f)}\,a(f)\,df + \mathcal{O}(\epsilon^2).
\end{aligned}
\end{equation}
To interpret this geometrically, we define the \emph{normalized signal power distribution} $q(f)$:
\begin{equation}
  q(f) := \frac{|\tilde{h}(f)|^2/S_0(f)}{J_0/4}, \quad \int_{0}^{\infty} q(f)\,df = 1.
\end{equation}
This allows us to express the correction as a statistical average. Factorizing $J_0$:
\begin{equation}
  J_\epsilon = J_0 \left[ 1 - 2\epsilon \int_{0}^{\infty} a(f)q(f)\,df \right] + \mathcal{O}(\epsilon^2).
\end{equation}
Defining the $q$-weighted average drift $\langle a \rangle_q := \int a(f)q(f)\,df$, we obtain the fractional change in squared SNR:
\begin{equation}
  \boxed{
  \frac{J_\epsilon - J_0}{J_0} = -2\epsilon\,\langle a \rangle_q + \mathcal{O}(\epsilon^2).
  }
  \label{eq:snr-drift-app}
\end{equation}
This result confirms that, to first order, the sensitivity loss is determined by the overlap between the spectral drift $a(f)$ and the signal's sensitive bandwidth~\cite{zackayDetectingGravitationalWaves2021}. In the main text, we use Eq.~\eqref{eq:snr-drift-app} to validate numerical simulations, verifying that the observed peak drop matches this analytic prediction.

\bibliographystyle{apsrev4-2}
\bibliography{references}

\begin{thebibliography}{51}%
\makeatletter
\providecommand \@ifxundefined [1]{%
 \@ifx{#1\undefined}
}%
\providecommand \@ifnum [1]{%
 \ifnum #1\expandafter \@firstoftwo
 \else \expandafter \@secondoftwo
 \fi
}%
\providecommand \@ifx [1]{%
 \ifx #1\expandafter \@firstoftwo
 \else \expandafter \@secondoftwo
 \fi
}%
\providecommand \natexlab [1]{#1}%
\providecommand \enquote  [1]{``#1''}%
\providecommand \bibnamefont  [1]{#1}%
\providecommand \bibfnamefont [1]{#1}%
\providecommand \citenamefont [1]{#1}%
\providecommand \href@noop [0]{\@secondoftwo}%
\providecommand \href [0]{\begingroup \@sanitize@url \@href}%
\providecommand \@href[1]{\@@startlink{#1}\@@href}%
\providecommand \@@href[1]{\endgroup#1\@@endlink}%
\providecommand \@sanitize@url [0]{\catcode `\\12\catcode `\$12\catcode
  `\&12\catcode `\#12\catcode `\^12\catcode `\_12\catcode `\%12\relax}%
\providecommand \@@startlink[1]{}%
\providecommand \@@endlink[0]{}%
\providecommand \url  [0]{\begingroup\@sanitize@url \@url }%
\providecommand \@url [1]{\endgroup\@href {#1}{\urlprefix }}%
\providecommand \urlprefix  [0]{URL }%
\providecommand \Eprint [0]{\href }%
\providecommand \doibase [0]{https://doi.org/}%
\providecommand \selectlanguage [0]{\@gobble}%
\providecommand \bibinfo  [0]{\@secondoftwo}%
\providecommand \bibfield  [0]{\@secondoftwo}%
\providecommand \translation [1]{[#1]}%
\providecommand \BibitemOpen [0]{}%
\providecommand \bibitemStop [0]{}%
\providecommand \bibitemNoStop [0]{.\EOS\space}%
\providecommand \EOS [0]{\spacefactor3000\relax}%
\providecommand \BibitemShut  [1]{\csname bibitem#1\endcsname}%
\let\auto@bib@innerbib\@empty
\bibitem [{\citenamefont {{The LIGO Scientific Collaboration}}\ and\
  \citenamefont {{the Virgo
  Collaboration}}(2016)}]{theligoscientificcollaborationGW150914FirstResults2016}%
  \BibitemOpen
  \bibfield  {author} {\bibinfo {author} {\bibnamefont {{The LIGO Scientific
  Collaboration}}}\ and\ \bibinfo {author} {\bibnamefont {{the Virgo
  Collaboration}}},\ }\href {https://doi.org/10.1103/PhysRevD.93.122003}
  {\bibfield  {journal} {\bibinfo  {journal} {Physical Review D}\ }\textbf
  {\bibinfo {volume} {93}},\ \bibinfo {pages} {122003} (\bibinfo {year}
  {2016})},\ \Eprint {https://arxiv.org/abs/1602.03839} {arXiv:1602.03839
  [gr-qc]} \BibitemShut {NoStop}%
\bibitem [{\citenamefont {Allen}\ \emph {et~al.}(2012)\citenamefont {Allen},
  \citenamefont {Anderson}, \citenamefont {Brady}, \citenamefont {Brown},\ and\
  \citenamefont {Creighton}}]{allenFINDCHIRPAlgorithmDetection2012a}%
  \BibitemOpen
  \bibfield  {author} {\bibinfo {author} {\bibfnamefont {B.}~\bibnamefont
  {Allen}}, \bibinfo {author} {\bibfnamefont {W.~G.}\ \bibnamefont {Anderson}},
  \bibinfo {author} {\bibfnamefont {P.~R.}\ \bibnamefont {Brady}}, \bibinfo
  {author} {\bibfnamefont {D.~A.}\ \bibnamefont {Brown}},\ and\ \bibinfo
  {author} {\bibfnamefont {J.~D.~E.}\ \bibnamefont {Creighton}},\ }\href
  {https://doi.org/10.1103/PhysRevD.85.122006} {\bibfield  {journal} {\bibinfo
  {journal} {Physical Review D}\ }\textbf {\bibinfo {volume} {85}},\ \bibinfo
  {pages} {122006} (\bibinfo {year} {2012})},\ \Eprint
  {https://arxiv.org/abs/gr-qc/0509116} {arXiv:gr-qc/0509116} \BibitemShut
  {NoStop}%
\bibitem [{\citenamefont {Cutler}\ and\ \citenamefont
  {Flanagan}(1994)}]{cutlerGravitationalWavesMergin1994a}%
  \BibitemOpen
  \bibfield  {author} {\bibinfo {author} {\bibfnamefont {C.}~\bibnamefont
  {Cutler}}\ and\ \bibinfo {author} {\bibfnamefont {E.}~\bibnamefont
  {Flanagan}},\ }\href {https://doi.org/10.1103/PhysRevD.49.2658} {\bibfield
  {journal} {\bibinfo  {journal} {Physical Review D}\ }\textbf {\bibinfo
  {volume} {49}},\ \bibinfo {pages} {2658} (\bibinfo {year} {1994})},\ \Eprint
  {https://arxiv.org/abs/gr-qc/9402014} {arXiv:gr-qc/9402014} \BibitemShut
  {NoStop}%
\bibitem [{\citenamefont {Wainstein}\ and\ \citenamefont
  {Zubakov}(1970)}]{wainsteinExtractionSignalsNoise1970}%
  \BibitemOpen
  \bibfield  {author} {\bibinfo {author} {\bibfnamefont {L.~A.}\ \bibnamefont
  {Wainstein}}\ and\ \bibinfo {author} {\bibfnamefont {V.~D.}\ \bibnamefont
  {Zubakov}},\ }\href@noop {} {\emph {\bibinfo {title} {Extraction of Signals
  from Noise}}},\ \bibinfo {edition} {republ}\ ed.,\ Dover Books on Physics and
  Mathematical Physics\ (\bibinfo  {publisher} {Dover Publ},\ \bibinfo
  {address} {New York, NY},\ \bibinfo {year} {1970})\BibitemShut {NoStop}%
\bibitem [{\citenamefont {Littenberg}\ and\ \citenamefont
  {Cornish}(2015)}]{littenbergBayesLineBayesianInference2015}%
  \BibitemOpen
  \bibfield  {author} {\bibinfo {author} {\bibfnamefont {T.~B.}\ \bibnamefont
  {Littenberg}}\ and\ \bibinfo {author} {\bibfnamefont {N.~J.}\ \bibnamefont
  {Cornish}},\ }\href {https://doi.org/10.1103/PhysRevD.91.084034} {\bibfield
  {journal} {\bibinfo  {journal} {Physical Review D}\ }\textbf {\bibinfo
  {volume} {91}},\ \bibinfo {pages} {084034} (\bibinfo {year} {2015})},\
  \Eprint {https://arxiv.org/abs/1410.3852} {arXiv:1410.3852 [gr-qc]}
  \BibitemShut {NoStop}%
\bibitem [{\citenamefont {Oppenheim}\ and\ \citenamefont
  {Schafer}(2010)}]{oppenheimDiscretetimeSignalProcessing2010}%
  \BibitemOpen
  \bibinfo {editor} {\bibfnamefont {A.~V.}\ \bibnamefont {Oppenheim}}\ and\
  \bibinfo {editor} {\bibfnamefont {R.~W.}\ \bibnamefont {Schafer}},\ eds.,\
  \href@noop {} {\emph {\bibinfo {title} {Discrete-Time Signal Processing}}},\
  \bibinfo {edition} {3rd}\ ed.,\ Prentice {{Hall}} Signal Processing Series\
  (\bibinfo  {publisher} {Pearson},\ \bibinfo {address} {Upper Saddle River
  Munich},\ \bibinfo {year} {2010})\BibitemShut {NoStop}%
\bibitem [{\citenamefont {Smith}(2008)}]{smithIntroductionDigitalFilters2008}%
  \BibitemOpen
  \bibfield  {author} {\bibinfo {author} {\bibfnamefont {J.~O.~I.}\
  \bibnamefont {Smith}},\ }\href@noop {} {\emph {\bibinfo {title} {Introduction
  to Digital Filters: With Audio Applications}}},\ \bibinfo {edition} {2nd}\
  ed.\ (\bibinfo  {publisher} {Stanford Univ},\ \bibinfo {address} {Stanford,
  Calif},\ \bibinfo {year} {2008})\BibitemShut {NoStop}%
\bibitem [{\citenamefont {{C. Messick, et
  al.}}(2017)}]{messickAnalysisFrameworkPrompt2017b}%
  \BibitemOpen
  \bibfield  {author} {\bibinfo {author} {\bibnamefont {{C. Messick, et
  al.}}},\ }\href {https://doi.org/10.1103/PhysRevD.95.042001} {\bibfield
  {journal} {\bibinfo  {journal} {Physical Review D}\ }\textbf {\bibinfo
  {volume} {95}},\ \bibinfo {pages} {042001} (\bibinfo {year} {2017})},\
  \Eprint {https://arxiv.org/abs/1604.04324} {arXiv:1604.04324 [astro-ph]}
  \BibitemShut {NoStop}%
\bibitem [{\citenamefont {Cannon}\ \emph {et~al.}(2012)\citenamefont {Cannon},
  \citenamefont {Cariou}, \citenamefont {Chapman}, \citenamefont
  {{Crispin-Ortuzar}}, \citenamefont {Fotopoulos}, \citenamefont {Frei},
  \citenamefont {Hanna}, \citenamefont {Kara}, \citenamefont {Keppel},
  \citenamefont {Liao}, \citenamefont {Privitera}, \citenamefont {Searle},
  \citenamefont {Singer},\ and\ \citenamefont
  {Weinstein}}]{cannonEarlyWarningDetectionGravitational2012d}%
  \BibitemOpen
  \bibfield  {author} {\bibinfo {author} {\bibfnamefont {K.}~\bibnamefont
  {Cannon}}, \bibinfo {author} {\bibfnamefont {R.}~\bibnamefont {Cariou}},
  \bibinfo {author} {\bibfnamefont {A.}~\bibnamefont {Chapman}}, \bibinfo
  {author} {\bibfnamefont {M.}~\bibnamefont {{Crispin-Ortuzar}}}, \bibinfo
  {author} {\bibfnamefont {N.}~\bibnamefont {Fotopoulos}}, \bibinfo {author}
  {\bibfnamefont {M.}~\bibnamefont {Frei}}, \bibinfo {author} {\bibfnamefont
  {C.}~\bibnamefont {Hanna}}, \bibinfo {author} {\bibfnamefont
  {E.}~\bibnamefont {Kara}}, \bibinfo {author} {\bibfnamefont {D.}~\bibnamefont
  {Keppel}}, \bibinfo {author} {\bibfnamefont {L.}~\bibnamefont {Liao}},
  \bibinfo {author} {\bibfnamefont {S.}~\bibnamefont {Privitera}}, \bibinfo
  {author} {\bibfnamefont {A.}~\bibnamefont {Searle}}, \bibinfo {author}
  {\bibfnamefont {L.}~\bibnamefont {Singer}},\ and\ \bibinfo {author}
  {\bibfnamefont {A.}~\bibnamefont {Weinstein}},\ }\href
  {https://doi.org/10.1088/0004-637X/748/2/136} {\bibfield  {journal} {\bibinfo
   {journal} {The Astrophysical Journal}\ }\textbf {\bibinfo {volume} {748}},\
  \bibinfo {pages} {136} (\bibinfo {year} {2012})},\ \Eprint
  {https://arxiv.org/abs/1107.2665} {arXiv:1107.2665 [astro-ph]} \BibitemShut
  {NoStop}%
\bibitem [{\citenamefont {Nitz}\ \emph
  {et~al.}(2018{\natexlab{a}})\citenamefont {Nitz}, \citenamefont {Canton},
  \citenamefont {Davis},\ and\ \citenamefont {Reyes}}]{nitzPyCBCLiveRapid2018}%
  \BibitemOpen
  \bibfield  {author} {\bibinfo {author} {\bibfnamefont {A.~H.}\ \bibnamefont
  {Nitz}}, \bibinfo {author} {\bibfnamefont {T.~D.}\ \bibnamefont {Canton}},
  \bibinfo {author} {\bibfnamefont {D.}~\bibnamefont {Davis}},\ and\ \bibinfo
  {author} {\bibfnamefont {S.}~\bibnamefont {Reyes}},\ }\href
  {https://doi.org/10.1103/PhysRevD.98.024050} {\bibfield  {journal} {\bibinfo
  {journal} {Physical Review D}\ }\textbf {\bibinfo {volume} {98}},\ \bibinfo
  {pages} {024050} (\bibinfo {year} {2018}{\natexlab{a}})},\ \Eprint
  {https://arxiv.org/abs/1805.11174} {arXiv:1805.11174 [gr-qc]} \BibitemShut
  {NoStop}%
\bibitem [{\citenamefont {Adams}\ \emph {et~al.}(2016)\citenamefont {Adams},
  \citenamefont {Buskulic}, \citenamefont {Germain}, \citenamefont {Guidi},
  \citenamefont {Marion}, \citenamefont {Montani}, \citenamefont {Mours},
  \citenamefont {Piergiovanni},\ and\ \citenamefont
  {Wang}}]{adamsLowlatencyAnalysisPipeline2016a}%
  \BibitemOpen
  \bibfield  {author} {\bibinfo {author} {\bibfnamefont {T.}~\bibnamefont
  {Adams}}, \bibinfo {author} {\bibfnamefont {D.}~\bibnamefont {Buskulic}},
  \bibinfo {author} {\bibfnamefont {V.}~\bibnamefont {Germain}}, \bibinfo
  {author} {\bibfnamefont {G.~M.}\ \bibnamefont {Guidi}}, \bibinfo {author}
  {\bibfnamefont {F.}~\bibnamefont {Marion}}, \bibinfo {author} {\bibfnamefont
  {M.}~\bibnamefont {Montani}}, \bibinfo {author} {\bibfnamefont
  {B.}~\bibnamefont {Mours}}, \bibinfo {author} {\bibfnamefont
  {F.}~\bibnamefont {Piergiovanni}},\ and\ \bibinfo {author} {\bibfnamefont
  {G.}~\bibnamefont {Wang}},\ }\href
  {https://doi.org/10.1088/0264-9381/33/17/175012} {\bibfield  {journal}
  {\bibinfo  {journal} {Classical and Quantum Gravity}\ }\textbf {\bibinfo
  {volume} {33}},\ \bibinfo {pages} {175012} (\bibinfo {year} {2016})},\
  \Eprint {https://arxiv.org/abs/1512.02864} {arXiv:1512.02864 [gr-qc]}
  \BibitemShut {NoStop}%
\bibitem [{\citenamefont {Tsukada}\ \emph {et~al.}(2018)\citenamefont
  {Tsukada}, \citenamefont {Cannon}, \citenamefont {Hanna}, \citenamefont
  {Keppel}, \citenamefont {Meacher},\ and\ \citenamefont
  {Messick}}]{tsukadaApplicationZerolatencyWhitening2018}%
  \BibitemOpen
  \bibfield  {author} {\bibinfo {author} {\bibfnamefont {L.}~\bibnamefont
  {Tsukada}}, \bibinfo {author} {\bibfnamefont {K.}~\bibnamefont {Cannon}},
  \bibinfo {author} {\bibfnamefont {C.}~\bibnamefont {Hanna}}, \bibinfo
  {author} {\bibfnamefont {D.}~\bibnamefont {Keppel}}, \bibinfo {author}
  {\bibfnamefont {D.}~\bibnamefont {Meacher}},\ and\ \bibinfo {author}
  {\bibfnamefont {C.}~\bibnamefont {Messick}},\ }\href
  {https://doi.org/10.1103/PhysRevD.97.103009} {\bibfield  {journal} {\bibinfo
  {journal} {Physical Review D}\ }\textbf {\bibinfo {volume} {97}},\ \bibinfo
  {pages} {103009} (\bibinfo {year} {2018})},\ \Eprint
  {https://arxiv.org/abs/1708.04125} {arXiv:1708.04125 [astro-ph]} \BibitemShut
  {NoStop}%
\bibitem [{\citenamefont {{B. Ewing et
  al.}}(2024)}]{ewingPerformanceLowlatencyGstLAL2024}%
  \BibitemOpen
  \bibfield  {author} {\bibinfo {author} {\bibnamefont {{B. Ewing et al.}}},\
  }\href {https://doi.org/10.1103/PhysRevD.109.042008} {\bibfield  {journal}
  {\bibinfo  {journal} {Physical Review D}\ }\textbf {\bibinfo {volume}
  {109}},\ \bibinfo {pages} {042008} (\bibinfo {year} {2024})}\BibitemShut
  {NoStop}%
\bibitem [{\citenamefont {Nitz}\ \emph
  {et~al.}(2018{\natexlab{b}})\citenamefont {Nitz}, \citenamefont {Dal~Canton},
  \citenamefont {Davis},\ and\ \citenamefont
  {Reyes}}]{nitzRapidDetectionGravitational2018}%
  \BibitemOpen
  \bibfield  {author} {\bibinfo {author} {\bibfnamefont {A.~H.}\ \bibnamefont
  {Nitz}}, \bibinfo {author} {\bibfnamefont {T.}~\bibnamefont {Dal~Canton}},
  \bibinfo {author} {\bibfnamefont {D.}~\bibnamefont {Davis}},\ and\ \bibinfo
  {author} {\bibfnamefont {S.}~\bibnamefont {Reyes}},\ }\href
  {https://doi.org/10.1103/PhysRevD.98.024050} {\bibfield  {journal} {\bibinfo
  {journal} {Physical Review D}\ }\textbf {\bibinfo {volume} {98}},\ \bibinfo
  {pages} {024050} (\bibinfo {year} {2018}{\natexlab{b}})}\BibitemShut
  {NoStop}%
\bibitem [{\citenamefont {Chu}\ \emph {et~al.}(2021)\citenamefont {Chu},
  \citenamefont {Kovalam}, \citenamefont {Wen}, \citenamefont {{Slaven-Blair}},
  \citenamefont {Bosveld}, \citenamefont {Chen}, \citenamefont {Clearwater},
  \citenamefont {Codoreanu}, \citenamefont {Du}, \citenamefont {Guo},
  \citenamefont {Guo}, \citenamefont {Kim}, \citenamefont {Li}, \citenamefont
  {Oloworaran}, \citenamefont {Panther}, \citenamefont {Powell}, \citenamefont
  {Sengupta}, \citenamefont {Wette},\ and\ \citenamefont
  {Zhu}}]{chuSPIIROnlineCoherent2021}%
  \BibitemOpen
  \bibfield  {author} {\bibinfo {author} {\bibfnamefont {Q.}~\bibnamefont
  {Chu}}, \bibinfo {author} {\bibfnamefont {M.}~\bibnamefont {Kovalam}},
  \bibinfo {author} {\bibfnamefont {L.}~\bibnamefont {Wen}}, \bibinfo {author}
  {\bibfnamefont {T.}~\bibnamefont {{Slaven-Blair}}}, \bibinfo {author}
  {\bibfnamefont {J.}~\bibnamefont {Bosveld}}, \bibinfo {author} {\bibfnamefont
  {Y.}~\bibnamefont {Chen}}, \bibinfo {author} {\bibfnamefont {P.}~\bibnamefont
  {Clearwater}}, \bibinfo {author} {\bibfnamefont {A.}~\bibnamefont
  {Codoreanu}}, \bibinfo {author} {\bibfnamefont {Z.}~\bibnamefont {Du}},
  \bibinfo {author} {\bibfnamefont {X.}~\bibnamefont {Guo}}, \bibinfo {author}
  {\bibfnamefont {X.}~\bibnamefont {Guo}}, \bibinfo {author} {\bibfnamefont
  {K.}~\bibnamefont {Kim}}, \bibinfo {author} {\bibfnamefont {T.~G.~F.}\
  \bibnamefont {Li}}, \bibinfo {author} {\bibfnamefont {V.}~\bibnamefont
  {Oloworaran}}, \bibinfo {author} {\bibfnamefont {F.}~\bibnamefont {Panther}},
  \bibinfo {author} {\bibfnamefont {J.}~\bibnamefont {Powell}}, \bibinfo
  {author} {\bibfnamefont {A.~S.}\ \bibnamefont {Sengupta}}, \bibinfo {author}
  {\bibfnamefont {K.}~\bibnamefont {Wette}},\ and\ \bibinfo {author}
  {\bibfnamefont {X.}~\bibnamefont {Zhu}},\ }\href
  {https://doi.org/10.48550/arXiv.2011.06787} {\bibinfo {title} {The {{SPIIR}}
  online coherent pipeline to search for gravitational waves from compact
  binary coalescences}} (\bibinfo {year} {2021}),\ \Eprint
  {https://arxiv.org/abs/2011.06787} {arXiv:2011.06787 [gr-qc]} \BibitemShut
  {NoStop}%
\bibitem [{\citenamefont {All{\'e}n{\'e}}\ \emph {et~al.}(2025)\citenamefont
  {All{\'e}n{\'e}}, \citenamefont {Aubin}, \citenamefont {Bentara},
  \citenamefont {Buskulic}, \citenamefont {Guidi}, \citenamefont {Juste},
  \citenamefont {Lethuillier}, \citenamefont {Marion}, \citenamefont {Mobilia},
  \citenamefont {Mours}, \citenamefont {Ouzriat}, \citenamefont {Sainrat},\
  and\ \citenamefont {Sordini}}]{alleneMBTAPipelineDetecting2025}%
  \BibitemOpen
  \bibfield  {author} {\bibinfo {author} {\bibfnamefont {C.}~\bibnamefont
  {All{\'e}n{\'e}}}, \bibinfo {author} {\bibfnamefont {F.}~\bibnamefont
  {Aubin}}, \bibinfo {author} {\bibfnamefont {I.}~\bibnamefont {Bentara}},
  \bibinfo {author} {\bibfnamefont {D.}~\bibnamefont {Buskulic}}, \bibinfo
  {author} {\bibfnamefont {G.~M.}\ \bibnamefont {Guidi}}, \bibinfo {author}
  {\bibfnamefont {V.}~\bibnamefont {Juste}}, \bibinfo {author} {\bibfnamefont
  {M.}~\bibnamefont {Lethuillier}}, \bibinfo {author} {\bibfnamefont
  {F.}~\bibnamefont {Marion}}, \bibinfo {author} {\bibfnamefont
  {L.}~\bibnamefont {Mobilia}}, \bibinfo {author} {\bibfnamefont
  {B.}~\bibnamefont {Mours}}, \bibinfo {author} {\bibfnamefont
  {A.}~\bibnamefont {Ouzriat}}, \bibinfo {author} {\bibfnamefont
  {T.}~\bibnamefont {Sainrat}},\ and\ \bibinfo {author} {\bibfnamefont
  {V.}~\bibnamefont {Sordini}},\ }\href
  {https://doi.org/10.1088/1361-6382/add234} {\bibfield  {journal} {\bibinfo
  {journal} {Classical and Quantum Gravity}\ }\textbf {\bibinfo {volume}
  {42}},\ \bibinfo {pages} {105009} (\bibinfo {year} {2025})},\ \Eprint
  {https://arxiv.org/abs/2501.04598} {arXiv:2501.04598 [gr-qc]} \BibitemShut
  {NoStop}%
\bibitem [{\citenamefont {{A. Buikema et
  al.}}(2020)}]{buikemaSensitivityPerformanceAdvanced2020}%
  \BibitemOpen
  \bibfield  {author} {\bibinfo {author} {\bibnamefont {{A. Buikema et al.}}},\
  }\href {https://doi.org/10.1103/PhysRevD.102.062003} {\bibfield  {journal}
  {\bibinfo  {journal} {Physical Review D}\ }\textbf {\bibinfo {volume}
  {102}},\ \bibinfo {pages} {062003} (\bibinfo {year} {2020})}\BibitemShut
  {NoStop}%
\bibitem [{\citenamefont {{P. B. Covas et
  al.}}(2018)}]{covasIdentificationMitigationNarrow2018}%
  \BibitemOpen
  \bibfield  {author} {\bibinfo {author} {\bibnamefont {{P. B. Covas et
  al.}}},\ }\href {https://doi.org/10.1103/PhysRevD.97.082002} {\bibfield
  {journal} {\bibinfo  {journal} {Physical Review D}\ }\textbf {\bibinfo
  {volume} {97}},\ \bibinfo {pages} {082002} (\bibinfo {year} {2018})},\
  \Eprint {https://arxiv.org/abs/1801.07204} {arXiv:1801.07204 [astro-ph]}
  \BibitemShut {NoStop}%
\bibitem [{\citenamefont {Magee}\ \emph {et~al.}(2024)\citenamefont {Magee},
  \citenamefont {Isi}, \citenamefont {Payne}, \citenamefont {Chatziioannou},
  \citenamefont {Farr}, \citenamefont {Pratten},\ and\ \citenamefont
  {Vitale}}]{mageeImpactSelectionBiases2024}%
  \BibitemOpen
  \bibfield  {author} {\bibinfo {author} {\bibfnamefont {R.}~\bibnamefont
  {Magee}}, \bibinfo {author} {\bibfnamefont {M.}~\bibnamefont {Isi}}, \bibinfo
  {author} {\bibfnamefont {E.}~\bibnamefont {Payne}}, \bibinfo {author}
  {\bibfnamefont {K.}~\bibnamefont {Chatziioannou}}, \bibinfo {author}
  {\bibfnamefont {W.~M.}\ \bibnamefont {Farr}}, \bibinfo {author}
  {\bibfnamefont {G.}~\bibnamefont {Pratten}},\ and\ \bibinfo {author}
  {\bibfnamefont {S.}~\bibnamefont {Vitale}},\ }\href
  {https://doi.org/10.1103/PhysRevD.109.023014} {\bibfield  {journal} {\bibinfo
   {journal} {Physical Review D}\ }\textbf {\bibinfo {volume} {109}},\ \bibinfo
  {pages} {023014} (\bibinfo {year} {2024})}\BibitemShut {NoStop}%
\bibitem [{\citenamefont {Littenberg}\ and\ \citenamefont
  {Cornish}(2010)}]{littenbergSeparatingGravitationalWave2010}%
  \BibitemOpen
  \bibfield  {author} {\bibinfo {author} {\bibfnamefont {T.~B.}\ \bibnamefont
  {Littenberg}}\ and\ \bibinfo {author} {\bibfnamefont {N.~J.}\ \bibnamefont
  {Cornish}},\ }\href {https://doi.org/10.1103/PhysRevD.82.103007} {\bibfield
  {journal} {\bibinfo  {journal} {Physical Review D}\ }\textbf {\bibinfo
  {volume} {82}},\ \bibinfo {pages} {103007} (\bibinfo {year} {2010})},\
  \Eprint {https://arxiv.org/abs/1008.1577} {arXiv:1008.1577 [gr-qc]}
  \BibitemShut {NoStop}%
\bibitem [{\citenamefont {Cutler}\ and\ \citenamefont
  {Vallisneri}(2007)}]{cutlerLISADetectionsMassive2007}%
  \BibitemOpen
  \bibfield  {author} {\bibinfo {author} {\bibfnamefont {C.}~\bibnamefont
  {Cutler}}\ and\ \bibinfo {author} {\bibfnamefont {M.}~\bibnamefont
  {Vallisneri}},\ }\href {https://doi.org/10.1103/PhysRevD.76.104018}
  {\bibfield  {journal} {\bibinfo  {journal} {Physical Review D}\ }\textbf
  {\bibinfo {volume} {76}},\ \bibinfo {pages} {104018} (\bibinfo {year}
  {2007})},\ \Eprint {https://arxiv.org/abs/0707.2982} {arXiv:0707.2982
  [gr-qc]} \BibitemShut {NoStop}%
\bibitem [{\citenamefont {Owen}(1996)}]{owenSearchTemplatesGravitational1996a}%
  \BibitemOpen
  \bibfield  {author} {\bibinfo {author} {\bibfnamefont {B.~J.}\ \bibnamefont
  {Owen}},\ }\href {https://doi.org/10.1103/PhysRevD.53.6749} {\bibfield
  {journal} {\bibinfo  {journal} {Physical Review D}\ }\textbf {\bibinfo
  {volume} {53}},\ \bibinfo {pages} {6749} (\bibinfo {year} {1996})},\ \Eprint
  {https://arxiv.org/abs/gr-qc/9511032} {arXiv:gr-qc/9511032} \BibitemShut
  {NoStop}%
\bibitem [{\citenamefont {{The LIGO Scientific Collaboration}}\ \emph
  {et~al.}(2025{\natexlab{a}})\citenamefont {{The LIGO Scientific
  Collaboration}}, \citenamefont {{the Virgo Collaboration}},\ and\
  \citenamefont {{the KAGRA
  Collaboration}}}]{collaborationGWTC40UpdatingGravitationalWave2025}%
  \BibitemOpen
  \bibfield  {author} {\bibinfo {author} {\bibnamefont {{The LIGO Scientific
  Collaboration}}}, \bibinfo {author} {\bibnamefont {{the Virgo
  Collaboration}}},\ and\ \bibinfo {author} {\bibnamefont {{the KAGRA
  Collaboration}}},\ }\href {https://doi.org/10.48550/arXiv.2508.18082}
  {\bibinfo {title} {{{GWTC-4}}.0: {{Updating}} the {{Gravitational-Wave
  Transient Catalog}} with {{Observations}} from the {{First Part}} of the
  {{Fourth LIGO-Virgo-KAGRA Observing Run}}}} (\bibinfo {year}
  {2025}{\natexlab{a}}),\ \Eprint {https://arxiv.org/abs/2508.18082}
  {arXiv:2508.18082 [gr-qc]} \BibitemShut {NoStop}%
\bibitem [{\citenamefont {{The LIGO Scientific Collaboration}}\ \emph
  {et~al.}(2025{\natexlab{b}})\citenamefont {{The LIGO Scientific
  Collaboration}}, \citenamefont {{the Virgo Collaboration}},\ and\
  \citenamefont {{the KAGRA Collaboration}}}]{collaborationOpenDataLIGO2025}%
  \BibitemOpen
  \bibfield  {author} {\bibinfo {author} {\bibnamefont {{The LIGO Scientific
  Collaboration}}}, \bibinfo {author} {\bibnamefont {{the Virgo
  Collaboration}}},\ and\ \bibinfo {author} {\bibnamefont {{the KAGRA
  Collaboration}}},\ }\href {https://doi.org/10.48550/arXiv.2508.18079}
  {\bibinfo {title} {Open {{Data}} from {{LIGO}}, {{Virgo}}, and {{KAGRA}}
  through the {{First Part}} of the {{Fourth Observing Run}}}} (\bibinfo {year}
  {2025}{\natexlab{b}}),\ \Eprint {https://arxiv.org/abs/2508.18079}
  {arXiv:2508.18079 [gr-qc]} \BibitemShut {NoStop}%
\bibitem [{\citenamefont {Hanna}\ \emph {et~al.}(2022)\citenamefont {Hanna},
  \citenamefont {Kennington}, \citenamefont {Sakon}, \citenamefont {Privitera},
  \citenamefont {Fernandez}, \citenamefont {Wang}, \citenamefont {Messick},
  \citenamefont {Pace}, \citenamefont {Cannon}, \citenamefont {Joshi},
  \citenamefont {Huxford}, \citenamefont {Caudill}, \citenamefont {Chan},
  \citenamefont {Cousins}, \citenamefont {Creighton}, \citenamefont {Ewing},
  \citenamefont {Fong}, \citenamefont {Godwin}, \citenamefont {Magee},
  \citenamefont {Meacher}, \citenamefont {Morisaki}, \citenamefont {Mukherjee},
  \citenamefont {Ohta}, \citenamefont {Sachdev}, \citenamefont {Singh},
  \citenamefont {Tapia}, \citenamefont {Tsukada}, \citenamefont {Tsuna},
  \citenamefont {Tsutsui}, \citenamefont {Ueno}, \citenamefont {Viets},
  \citenamefont {Wade},\ and\ \citenamefont
  {Wade}}]{hannaBinaryTreeApproach2022c}%
  \BibitemOpen
  \bibfield  {author} {\bibinfo {author} {\bibfnamefont {C.}~\bibnamefont
  {Hanna}}, \bibinfo {author} {\bibfnamefont {J.}~\bibnamefont {Kennington}},
  \bibinfo {author} {\bibfnamefont {S.}~\bibnamefont {Sakon}}, \bibinfo
  {author} {\bibfnamefont {S.}~\bibnamefont {Privitera}}, \bibinfo {author}
  {\bibfnamefont {M.}~\bibnamefont {Fernandez}}, \bibinfo {author}
  {\bibfnamefont {J.}~\bibnamefont {Wang}}, \bibinfo {author} {\bibfnamefont
  {C.}~\bibnamefont {Messick}}, \bibinfo {author} {\bibfnamefont
  {A.}~\bibnamefont {Pace}}, \bibinfo {author} {\bibfnamefont {K.}~\bibnamefont
  {Cannon}}, \bibinfo {author} {\bibfnamefont {P.}~\bibnamefont {Joshi}},
  \bibinfo {author} {\bibfnamefont {R.}~\bibnamefont {Huxford}}, \bibinfo
  {author} {\bibfnamefont {S.}~\bibnamefont {Caudill}}, \bibinfo {author}
  {\bibfnamefont {C.}~\bibnamefont {Chan}}, \bibinfo {author} {\bibfnamefont
  {B.}~\bibnamefont {Cousins}}, \bibinfo {author} {\bibfnamefont {J.~D.~E.}\
  \bibnamefont {Creighton}}, \bibinfo {author} {\bibfnamefont {B.}~\bibnamefont
  {Ewing}}, \bibinfo {author} {\bibfnamefont {H.}~\bibnamefont {Fong}},
  \bibinfo {author} {\bibfnamefont {P.}~\bibnamefont {Godwin}}, \bibinfo
  {author} {\bibfnamefont {R.}~\bibnamefont {Magee}}, \bibinfo {author}
  {\bibfnamefont {D.}~\bibnamefont {Meacher}}, \bibinfo {author} {\bibfnamefont
  {S.}~\bibnamefont {Morisaki}}, \bibinfo {author} {\bibfnamefont
  {D.}~\bibnamefont {Mukherjee}}, \bibinfo {author} {\bibfnamefont
  {H.}~\bibnamefont {Ohta}}, \bibinfo {author} {\bibfnamefont {S.}~\bibnamefont
  {Sachdev}}, \bibinfo {author} {\bibfnamefont {D.}~\bibnamefont {Singh}},
  \bibinfo {author} {\bibfnamefont {R.}~\bibnamefont {Tapia}}, \bibinfo
  {author} {\bibfnamefont {L.}~\bibnamefont {Tsukada}}, \bibinfo {author}
  {\bibfnamefont {D.}~\bibnamefont {Tsuna}}, \bibinfo {author} {\bibfnamefont
  {T.}~\bibnamefont {Tsutsui}}, \bibinfo {author} {\bibfnamefont
  {K.}~\bibnamefont {Ueno}}, \bibinfo {author} {\bibfnamefont {A.}~\bibnamefont
  {Viets}}, \bibinfo {author} {\bibfnamefont {L.}~\bibnamefont {Wade}},\ and\
  \bibinfo {author} {\bibfnamefont {M.}~\bibnamefont {Wade}},\ }\href
  {https://doi.org/10.48550/arXiv.2209.11298} {\bibinfo {title} {A binary tree
  approach to template placement for searches for gravitational waves from
  compact binary mergers}} (\bibinfo {year} {2022}),\ \Eprint
  {https://arxiv.org/abs/2209.11298} {arXiv:2209.11298 [gr-qc]} \BibitemShut
  {NoStop}%
\bibitem [{\citenamefont
  {Fairhurst}(2009)}]{fairhurstTriangulationGravitationalWave2009}%
  \BibitemOpen
  \bibfield  {author} {\bibinfo {author} {\bibfnamefont {S.}~\bibnamefont
  {Fairhurst}},\ }\href {https://doi.org/10.1088/1367-2630/11/12/123006
  10.1088/1367-2630/13/6/069602} {\bibfield  {journal} {\bibinfo  {journal}
  {New Journal of Physics}\ }\textbf {\bibinfo {volume} {11}},\ \bibinfo
  {pages} {123006} (\bibinfo {year} {2009})},\ \Eprint
  {https://arxiv.org/abs/0908.2356} {arXiv:0908.2356 [gr-qc]} \BibitemShut
  {NoStop}%
\bibitem [{\citenamefont {Sullivan}\ \emph {et~al.}(2023)\citenamefont
  {Sullivan}, \citenamefont {Asali}, \citenamefont {M{\'a}rka}, \citenamefont
  {Sigg}, \citenamefont {Countryman}, \citenamefont {Bartos}, \citenamefont
  {Kawabe}, \citenamefont {Pirello}, \citenamefont {Thomas}, \citenamefont
  {Shaffer}, \citenamefont {Thorne}, \citenamefont {Laxen}, \citenamefont
  {Betzwieser}, \citenamefont {Izumi}, \citenamefont {Bork}, \citenamefont
  {Ivanov}, \citenamefont {Barker}, \citenamefont {Adams}, \citenamefont
  {Clara}, \citenamefont {Factourovich},\ and\ \citenamefont
  {M{\'a}rka}}]{sullivanTimingSystemLIGO2023}%
  \BibitemOpen
  \bibfield  {author} {\bibinfo {author} {\bibfnamefont {A.~G.}\ \bibnamefont
  {Sullivan}}, \bibinfo {author} {\bibfnamefont {Y.}~\bibnamefont {Asali}},
  \bibinfo {author} {\bibfnamefont {Z.}~\bibnamefont {M{\'a}rka}}, \bibinfo
  {author} {\bibfnamefont {D.}~\bibnamefont {Sigg}}, \bibinfo {author}
  {\bibfnamefont {S.}~\bibnamefont {Countryman}}, \bibinfo {author}
  {\bibfnamefont {I.}~\bibnamefont {Bartos}}, \bibinfo {author} {\bibfnamefont
  {K.}~\bibnamefont {Kawabe}}, \bibinfo {author} {\bibfnamefont {M.~D.}\
  \bibnamefont {Pirello}}, \bibinfo {author} {\bibfnamefont {M.}~\bibnamefont
  {Thomas}}, \bibinfo {author} {\bibfnamefont {T.~J.}\ \bibnamefont {Shaffer}},
  \bibinfo {author} {\bibfnamefont {K.}~\bibnamefont {Thorne}}, \bibinfo
  {author} {\bibfnamefont {M.}~\bibnamefont {Laxen}}, \bibinfo {author}
  {\bibfnamefont {J.}~\bibnamefont {Betzwieser}}, \bibinfo {author}
  {\bibfnamefont {K.}~\bibnamefont {Izumi}}, \bibinfo {author} {\bibfnamefont
  {R.}~\bibnamefont {Bork}}, \bibinfo {author} {\bibfnamefont {A.}~\bibnamefont
  {Ivanov}}, \bibinfo {author} {\bibfnamefont {D.}~\bibnamefont {Barker}},
  \bibinfo {author} {\bibfnamefont {C.}~\bibnamefont {Adams}}, \bibinfo
  {author} {\bibfnamefont {F.}~\bibnamefont {Clara}}, \bibinfo {author}
  {\bibfnamefont {M.}~\bibnamefont {Factourovich}},\ and\ \bibinfo {author}
  {\bibfnamefont {S.}~\bibnamefont {M{\'a}rka}},\ }\href
  {https://doi.org/10.1103/PhysRevD.108.022003} {\bibfield  {journal} {\bibinfo
   {journal} {Physical Review D}\ }\textbf {\bibinfo {volume} {108}},\ \bibinfo
  {pages} {022003} (\bibinfo {year} {2023})},\ \Eprint
  {https://arxiv.org/abs/2304.01188} {arXiv:2304.01188 [astro-ph]} \BibitemShut
  {NoStop}%
\bibitem [{\citenamefont
  {Fairhurst}(2018)}]{fairhurstLocalizationTransientGravitational2018}%
  \BibitemOpen
  \bibfield  {author} {\bibinfo {author} {\bibfnamefont {S.}~\bibnamefont
  {Fairhurst}},\ }\href {https://doi.org/10.1088/1361-6382/aab675} {\bibfield
  {journal} {\bibinfo  {journal} {Classical and Quantum Gravity}\ }\textbf
  {\bibinfo {volume} {35}},\ \bibinfo {pages} {105002} (\bibinfo {year}
  {2018})},\ \Eprint {https://arxiv.org/abs/1712.04724} {arXiv:1712.04724
  [gr-qc]} \BibitemShut {NoStop}%
\bibitem [{\citenamefont {Zackay}\ \emph {et~al.}(2021)\citenamefont {Zackay},
  \citenamefont {Venumadhav}, \citenamefont {Roulet}, \citenamefont {Dai},\
  and\ \citenamefont {Zaldarriaga}}]{zackayDetectingGravitationalWaves2021}%
  \BibitemOpen
  \bibfield  {author} {\bibinfo {author} {\bibfnamefont {B.}~\bibnamefont
  {Zackay}}, \bibinfo {author} {\bibfnamefont {T.}~\bibnamefont {Venumadhav}},
  \bibinfo {author} {\bibfnamefont {J.}~\bibnamefont {Roulet}}, \bibinfo
  {author} {\bibfnamefont {L.}~\bibnamefont {Dai}},\ and\ \bibinfo {author}
  {\bibfnamefont {M.}~\bibnamefont {Zaldarriaga}},\ }\href
  {https://doi.org/10.1103/PhysRevD.104.063034} {\bibfield  {journal} {\bibinfo
   {journal} {Physical Review D}\ }\textbf {\bibinfo {volume} {104}},\ \bibinfo
  {pages} {063034} (\bibinfo {year} {2021})},\ \Eprint
  {https://arxiv.org/abs/1908.05644} {arXiv:1908.05644 [astro-ph]} \BibitemShut
  {NoStop}%
\bibitem [{\citenamefont {Biscoveanu}\ \emph {et~al.}(2020)\citenamefont
  {Biscoveanu}, \citenamefont {Haster}, \citenamefont {Vitale},\ and\
  \citenamefont {Davies}}]{biscoveanuQuantifyingEffectPower2020}%
  \BibitemOpen
  \bibfield  {author} {\bibinfo {author} {\bibfnamefont {S.}~\bibnamefont
  {Biscoveanu}}, \bibinfo {author} {\bibfnamefont {C.-J.}\ \bibnamefont
  {Haster}}, \bibinfo {author} {\bibfnamefont {S.}~\bibnamefont {Vitale}},\
  and\ \bibinfo {author} {\bibfnamefont {J.}~\bibnamefont {Davies}},\ }\href
  {https://doi.org/10.1103/PhysRevD.102.023008} {\bibfield  {journal} {\bibinfo
   {journal} {Physical Review D}\ }\textbf {\bibinfo {volume} {102}},\ \bibinfo
  {pages} {023008} (\bibinfo {year} {2020})},\ \Eprint
  {https://arxiv.org/abs/2004.05149} {arXiv:2004.05149 [astro-ph]} \BibitemShut
  {NoStop}%
\bibitem [{\citenamefont {Vitale}\ \emph {et~al.}(2012)\citenamefont {Vitale},
  \citenamefont {Pozzo}, \citenamefont {Li}, \citenamefont {Broeck},
  \citenamefont {Mandel}, \citenamefont {Aylott},\ and\ \citenamefont
  {Veitch}}]{vitaleEffectCalibrationErrors2012a}%
  \BibitemOpen
  \bibfield  {author} {\bibinfo {author} {\bibfnamefont {S.}~\bibnamefont
  {Vitale}}, \bibinfo {author} {\bibfnamefont {W.~D.}\ \bibnamefont {Pozzo}},
  \bibinfo {author} {\bibfnamefont {T.~G.~F.}\ \bibnamefont {Li}}, \bibinfo
  {author} {\bibfnamefont {C.~V.~D.}\ \bibnamefont {Broeck}}, \bibinfo {author}
  {\bibfnamefont {I.}~\bibnamefont {Mandel}}, \bibinfo {author} {\bibfnamefont
  {B.}~\bibnamefont {Aylott}},\ and\ \bibinfo {author} {\bibfnamefont
  {J.}~\bibnamefont {Veitch}},\ }\href
  {https://doi.org/10.1103/PhysRevD.85.064034} {\bibfield  {journal} {\bibinfo
  {journal} {Physical Review D}\ }\textbf {\bibinfo {volume} {85}},\ \bibinfo
  {pages} {064034} (\bibinfo {year} {2012})},\ \Eprint
  {https://arxiv.org/abs/1111.3044} {arXiv:1111.3044 [gr-qc]} \BibitemShut
  {NoStop}%
\bibitem [{\citenamefont {{J. C. Driggers et
  al.}}(2019)}]{driggersImprovingAstrophysicalParameter2019}%
  \BibitemOpen
  \bibfield  {author} {\bibinfo {author} {\bibnamefont {{J. C. Driggers et
  al.}}},\ }\href {https://doi.org/10.1103/PhysRevD.99.042001} {\bibfield
  {journal} {\bibinfo  {journal} {Physical Review D}\ }\textbf {\bibinfo
  {volume} {99}},\ \bibinfo {pages} {042001} (\bibinfo {year} {2019})},\
  \Eprint {https://arxiv.org/abs/1806.00532} {arXiv:1806.00532 [astro-ph]}
  \BibitemShut {NoStop}%
\bibitem [{\citenamefont {Davis}\ \emph {et~al.}(2019)\citenamefont {Davis},
  \citenamefont {Massinger}, \citenamefont {Lundgren}, \citenamefont
  {Driggers}, \citenamefont {Urban},\ and\ \citenamefont
  {Nuttall}}]{davisImprovingSensitivityAdvanced2019}%
  \BibitemOpen
  \bibfield  {author} {\bibinfo {author} {\bibfnamefont {D.}~\bibnamefont
  {Davis}}, \bibinfo {author} {\bibfnamefont {T.~J.}\ \bibnamefont
  {Massinger}}, \bibinfo {author} {\bibfnamefont {A.~P.}\ \bibnamefont
  {Lundgren}}, \bibinfo {author} {\bibfnamefont {J.~C.}\ \bibnamefont
  {Driggers}}, \bibinfo {author} {\bibfnamefont {A.~L.}\ \bibnamefont
  {Urban}},\ and\ \bibinfo {author} {\bibfnamefont {L.~K.}\ \bibnamefont
  {Nuttall}},\ }\href {https://doi.org/10.1088/1361-6382/ab01c5} {\bibfield
  {journal} {\bibinfo  {journal} {Classical and Quantum Gravity}\ }\textbf
  {\bibinfo {volume} {36}},\ \bibinfo {pages} {055011} (\bibinfo {year}
  {2019})},\ \Eprint {https://arxiv.org/abs/1809.05348} {arXiv:1809.05348
  [astro-ph]} \BibitemShut {NoStop}%
\bibitem [{\citenamefont {Boche}\ and\ \citenamefont
  {Pohl}(2005)}]{bocheSpectralFactorizationWhitening2005}%
  \BibitemOpen
  \bibfield  {author} {\bibinfo {author} {\bibfnamefont {H.}~\bibnamefont
  {Boche}}\ and\ \bibinfo {author} {\bibfnamefont {V.}~\bibnamefont {Pohl}},\
  }\href {https://doi.org/10.48550/arXiv.cs/0508018} {\bibinfo {title}
  {Spectral {{Factorization}}, {{Whitening-}} and {{Estimation Filter}} --
  {{Stability}}, {{Smoothness Properties}} and {{FIR Approximation Behavior}}}}
  (\bibinfo {year} {2005}),\ \Eprint {https://arxiv.org/abs/cs/0508018}
  {arXiv:cs/0508018} \BibitemShut {NoStop}%
\bibitem [{\citenamefont {{Damera-Venkata}}\ \emph {et~al.}(2000)\citenamefont
  {{Damera-Venkata}}, \citenamefont {Evans},\ and\ \citenamefont
  {McCaslin}}]{damera-venkataDesignOptimalMinimumphase2000}%
  \BibitemOpen
  \bibfield  {author} {\bibinfo {author} {\bibfnamefont {N.}~\bibnamefont
  {{Damera-Venkata}}}, \bibinfo {author} {\bibfnamefont {B.}~\bibnamefont
  {Evans}},\ and\ \bibinfo {author} {\bibfnamefont {S.}~\bibnamefont
  {McCaslin}},\ }\href {https://doi.org/10.1109/78.840000} {\bibfield
  {journal} {\bibinfo  {journal} {IEEE Transactions on Signal Processing}\
  }\textbf {\bibinfo {volume} {48}},\ \bibinfo {pages} {1491} (\bibinfo {year}
  {2000})}\BibitemShut {NoStop}%
\bibitem [{\citenamefont
  {Mecozzi}(2016)}]{mecozziNecessarySufficientCondition2016}%
  \BibitemOpen
  \bibfield  {author} {\bibinfo {author} {\bibfnamefont {A.}~\bibnamefont
  {Mecozzi}},\ }\href {https://doi.org/10.48550/arXiv.1606.04861} {\bibinfo
  {title} {A necessary and sufficient condition for minimum phase and
  implications for phase retrieval}} (\bibinfo {year} {2016}),\ \Eprint
  {https://arxiv.org/abs/1606.04861} {arXiv:1606.04861 [cs]} \BibitemShut
  {NoStop}%
\bibitem [{\citenamefont {Ephremidze}\ \emph {et~al.}(2016)\citenamefont
  {Ephremidze}, \citenamefont {Selesnick},\ and\ \citenamefont
  {Spitkovsky}}]{ephremidzeNonoptimalSpectralFactorizations2016}%
  \BibitemOpen
  \bibfield  {author} {\bibinfo {author} {\bibfnamefont {L.}~\bibnamefont
  {Ephremidze}}, \bibinfo {author} {\bibfnamefont {I.}~\bibnamefont
  {Selesnick}},\ and\ \bibinfo {author} {\bibfnamefont {I.}~\bibnamefont
  {Spitkovsky}},\ }\href {https://doi.org/10.48550/arXiv.1609.02058} {\bibinfo
  {title} {On non-optimal spectral factorizations}} (\bibinfo {year} {2016}),\
  \Eprint {https://arxiv.org/abs/1609.02058} {arXiv:1609.02058 [math]}
  \BibitemShut {NoStop}%
\bibitem [{\citenamefont {Ephremidze}\ \emph {et~al.}(2020)\citenamefont
  {Ephremidze}, \citenamefont {Shargorodsky},\ and\ \citenamefont
  {Spitkovsky}}]{ephremidzeQuantitativeResultsContinuity2020}%
  \BibitemOpen
  \bibfield  {author} {\bibinfo {author} {\bibfnamefont {L.}~\bibnamefont
  {Ephremidze}}, \bibinfo {author} {\bibfnamefont {E.}~\bibnamefont
  {Shargorodsky}},\ and\ \bibinfo {author} {\bibfnamefont {I.}~\bibnamefont
  {Spitkovsky}},\ }\href {https://doi.org/10.1112/jlms.12258} {\bibfield
  {journal} {\bibinfo  {journal} {Journal of the London Mathematical Society}\
  }\textbf {\bibinfo {volume} {101}},\ \bibinfo {pages} {60} (\bibinfo {year}
  {2020})},\ \Eprint {https://arxiv.org/abs/1804.00039} {arXiv:1804.00039
  [math]} \BibitemShut {NoStop}%
\bibitem [{\citenamefont {Droz}\ \emph {et~al.}(1999)\citenamefont {Droz},
  \citenamefont {Knapp}, \citenamefont {Poisson},\ and\ \citenamefont
  {Owen}}]{drozGravitationalWavesInspiraling1999}%
  \BibitemOpen
  \bibfield  {author} {\bibinfo {author} {\bibfnamefont {S.}~\bibnamefont
  {Droz}}, \bibinfo {author} {\bibfnamefont {D.~J.}\ \bibnamefont {Knapp}},
  \bibinfo {author} {\bibfnamefont {E.}~\bibnamefont {Poisson}},\ and\ \bibinfo
  {author} {\bibfnamefont {B.~J.}\ \bibnamefont {Owen}},\ }\href
  {https://doi.org/10.1103/PhysRevD.59.124016} {\bibfield  {journal} {\bibinfo
  {journal} {Physical Review D}\ }\textbf {\bibinfo {volume} {59}},\ \bibinfo
  {pages} {124016} (\bibinfo {year} {1999})},\ \Eprint
  {https://arxiv.org/abs/gr-qc/9901076} {arXiv:gr-qc/9901076} \BibitemShut
  {NoStop}%
\bibitem [{\citenamefont {{The LIGO Scientific Collaboration}}\ \emph
  {et~al.}(2020)\citenamefont {{The LIGO Scientific Collaboration}},
  \citenamefont {{the Virgo Collaboration}},\ and\ \citenamefont {{the KAGRA
  Collaboration}}}]{theligoscientificcollaborationProspectsObservingLocalizing2020}%
  \BibitemOpen
  \bibfield  {author} {\bibinfo {author} {\bibnamefont {{The LIGO Scientific
  Collaboration}}}, \bibinfo {author} {\bibnamefont {{the Virgo
  Collaboration}}},\ and\ \bibinfo {author} {\bibnamefont {{the KAGRA
  Collaboration}}},\ }\href {https://doi.org/10.1007/s41114-020-00026-9}
  {\bibfield  {journal} {\bibinfo  {journal} {Living Reviews in Relativity}\
  }\textbf {\bibinfo {volume} {23}},\ \bibinfo {pages} {3} (\bibinfo {year}
  {2020})},\ \Eprint {https://arxiv.org/abs/1304.0670} {arXiv:1304.0670
  [gr-qc]} \BibitemShut {NoStop}%
\bibitem [{\citenamefont {{The LIGO Scientific Collaboration}}\ and\
  \citenamefont {{the Virgo
  Collaboration}}(2020)}]{theligoscientificcollaborationGuideLIGOVirgoDetector2020}%
  \BibitemOpen
  \bibfield  {author} {\bibinfo {author} {\bibnamefont {{The LIGO Scientific
  Collaboration}}}\ and\ \bibinfo {author} {\bibnamefont {{the Virgo
  Collaboration}}},\ }\href {https://doi.org/10.1088/1361-6382/ab685e}
  {\bibfield  {journal} {\bibinfo  {journal} {Classical and Quantum Gravity}\
  }\textbf {\bibinfo {volume} {37}},\ \bibinfo {pages} {055002} (\bibinfo
  {year} {2020})},\ \Eprint {https://arxiv.org/abs/1908.11170}
  {arXiv:1908.11170 [gr-qc]} \BibitemShut {NoStop}%
\bibitem [{\citenamefont {Husa}\ \emph {et~al.}(2016)\citenamefont {Husa},
  \citenamefont {Khan}, \citenamefont {Hannam}, \citenamefont {P{\"u}rrer},
  \citenamefont {Ohme}, \citenamefont {Forteza},\ and\ \citenamefont
  {Boh{\'e}}}]{husaFrequencydomainGravitationalWaves2016}%
  \BibitemOpen
  \bibfield  {author} {\bibinfo {author} {\bibfnamefont {S.}~\bibnamefont
  {Husa}}, \bibinfo {author} {\bibfnamefont {S.}~\bibnamefont {Khan}}, \bibinfo
  {author} {\bibfnamefont {M.}~\bibnamefont {Hannam}}, \bibinfo {author}
  {\bibfnamefont {M.}~\bibnamefont {P{\"u}rrer}}, \bibinfo {author}
  {\bibfnamefont {F.}~\bibnamefont {Ohme}}, \bibinfo {author} {\bibfnamefont
  {X.~J.}\ \bibnamefont {Forteza}},\ and\ \bibinfo {author} {\bibfnamefont
  {A.}~\bibnamefont {Boh{\'e}}},\ }\href
  {https://doi.org/10.1103/PhysRevD.93.044006} {\bibfield  {journal} {\bibinfo
  {journal} {Physical Review D}\ }\textbf {\bibinfo {volume} {93}},\ \bibinfo
  {pages} {044006} (\bibinfo {year} {2016})},\ \Eprint
  {https://arxiv.org/abs/1508.07250} {arXiv:1508.07250 [gr-qc]} \BibitemShut
  {NoStop}%
\bibitem [{\citenamefont {Poisson}\ and\ \citenamefont
  {Will}(2014)}]{poissonGravityNewtonianPostNewtonian2014a}%
  \BibitemOpen
  \bibfield  {author} {\bibinfo {author} {\bibfnamefont {E.}~\bibnamefont
  {Poisson}}\ and\ \bibinfo {author} {\bibfnamefont {C.~M.}\ \bibnamefont
  {Will}},\ }\href {https://doi.org/10.1017/CBO9781139507486} {\emph {\bibinfo
  {title} {Gravity: {{Newtonian}}, Post-{{Newtonian}}, Relativistic}}},\
  \bibinfo {edition} {first published}\ ed.\ (\bibinfo  {publisher} {Cambridge
  University Press},\ \bibinfo {address} {Cambridge},\ \bibinfo {year}
  {2014})\BibitemShut {NoStop}%
\bibitem [{\citenamefont {Singer}\ and\ \citenamefont
  {Price}(2016)}]{singerRapidBayesianPosition2016}%
  \BibitemOpen
  \bibfield  {author} {\bibinfo {author} {\bibfnamefont {L.~P.}\ \bibnamefont
  {Singer}}\ and\ \bibinfo {author} {\bibfnamefont {L.~R.}\ \bibnamefont
  {Price}},\ }\href {https://doi.org/10.1103/PhysRevD.93.024013} {\bibfield
  {journal} {\bibinfo  {journal} {Physical Review D}\ }\textbf {\bibinfo
  {volume} {93}},\ \bibinfo {pages} {024013} (\bibinfo {year} {2016})},\
  \Eprint {https://arxiv.org/abs/1508.03634} {arXiv:1508.03634 [gr-qc]}
  \BibitemShut {NoStop}%
\bibitem [{\citenamefont {Nelder}\ and\ \citenamefont
  {Mead}(1965)}]{nelderSimplexMethodFunction1965}%
  \BibitemOpen
  \bibfield  {author} {\bibinfo {author} {\bibfnamefont {J.~A.}\ \bibnamefont
  {Nelder}}\ and\ \bibinfo {author} {\bibfnamefont {R.}~\bibnamefont {Mead}},\
  }\href {https://doi.org/10.1093/comjnl/7.4.308} {\bibfield  {journal}
  {\bibinfo  {journal} {The Computer Journal}\ }\textbf {\bibinfo {volume}
  {7}},\ \bibinfo {pages} {308} (\bibinfo {year} {1965})}\BibitemShut {NoStop}%
\bibitem [{\citenamefont {Chen}\ \emph {et~al.}(2021)\citenamefont {Chen},
  \citenamefont {Holz}, \citenamefont {Miller}, \citenamefont {Evans},
  \citenamefont {Vitale},\ and\ \citenamefont
  {Creighton}}]{chenDistanceMeasuresGravitationalwave2021}%
  \BibitemOpen
  \bibfield  {author} {\bibinfo {author} {\bibfnamefont {H.-Y.}\ \bibnamefont
  {Chen}}, \bibinfo {author} {\bibfnamefont {D.~E.}\ \bibnamefont {Holz}},
  \bibinfo {author} {\bibfnamefont {J.}~\bibnamefont {Miller}}, \bibinfo
  {author} {\bibfnamefont {M.}~\bibnamefont {Evans}}, \bibinfo {author}
  {\bibfnamefont {S.}~\bibnamefont {Vitale}},\ and\ \bibinfo {author}
  {\bibfnamefont {J.}~\bibnamefont {Creighton}},\ }\href
  {https://doi.org/10.1088/1361-6382/abd594} {\bibfield  {journal} {\bibinfo
  {journal} {Classical and Quantum Gravity}\ }\textbf {\bibinfo {volume}
  {38}},\ \bibinfo {pages} {055010} (\bibinfo {year} {2021})},\ \Eprint
  {https://arxiv.org/abs/1709.08079} {arXiv:1709.08079 [astro-ph]} \BibitemShut
  {NoStop}%
\bibitem [{\citenamefont {Sachdev}\ \emph {et~al.}(2019)\citenamefont
  {Sachdev}, \citenamefont {Caudill}, \citenamefont {Fong}, \citenamefont {Lo},
  \citenamefont {Messick}, \citenamefont {Mukherjee}, \citenamefont {Magee},
  \citenamefont {Tsukada}, \citenamefont {Blackburn}, \citenamefont {Brady},
  \citenamefont {Brockill}, \citenamefont {Cannon}, \citenamefont {Chamberlin},
  \citenamefont {Chatterjee}, \citenamefont {Creighton}, \citenamefont
  {Godwin}, \citenamefont {Gupta}, \citenamefont {Hanna}, \citenamefont
  {Kapadia}, \citenamefont {Lang}, \citenamefont {Li}, \citenamefont {Meacher},
  \citenamefont {Pace}, \citenamefont {Privitera}, \citenamefont {Sadeghian},
  \citenamefont {Wade}, \citenamefont {Wade}, \citenamefont {Weinstein},\ and\
  \citenamefont {Xiao}}]{sachdevGstLALSearchAnalysis2019a}%
  \BibitemOpen
  \bibfield  {author} {\bibinfo {author} {\bibfnamefont {S.}~\bibnamefont
  {Sachdev}}, \bibinfo {author} {\bibfnamefont {S.}~\bibnamefont {Caudill}},
  \bibinfo {author} {\bibfnamefont {H.}~\bibnamefont {Fong}}, \bibinfo {author}
  {\bibfnamefont {R.~K.~L.}\ \bibnamefont {Lo}}, \bibinfo {author}
  {\bibfnamefont {C.}~\bibnamefont {Messick}}, \bibinfo {author} {\bibfnamefont
  {D.}~\bibnamefont {Mukherjee}}, \bibinfo {author} {\bibfnamefont
  {R.}~\bibnamefont {Magee}}, \bibinfo {author} {\bibfnamefont
  {L.}~\bibnamefont {Tsukada}}, \bibinfo {author} {\bibfnamefont
  {K.}~\bibnamefont {Blackburn}}, \bibinfo {author} {\bibfnamefont
  {P.}~\bibnamefont {Brady}}, \bibinfo {author} {\bibfnamefont
  {P.}~\bibnamefont {Brockill}}, \bibinfo {author} {\bibfnamefont
  {K.}~\bibnamefont {Cannon}}, \bibinfo {author} {\bibfnamefont {S.~J.}\
  \bibnamefont {Chamberlin}}, \bibinfo {author} {\bibfnamefont
  {D.}~\bibnamefont {Chatterjee}}, \bibinfo {author} {\bibfnamefont {J.~D.~E.}\
  \bibnamefont {Creighton}}, \bibinfo {author} {\bibfnamefont {P.}~\bibnamefont
  {Godwin}}, \bibinfo {author} {\bibfnamefont {A.}~\bibnamefont {Gupta}},
  \bibinfo {author} {\bibfnamefont {C.}~\bibnamefont {Hanna}}, \bibinfo
  {author} {\bibfnamefont {S.}~\bibnamefont {Kapadia}}, \bibinfo {author}
  {\bibfnamefont {R.~N.}\ \bibnamefont {Lang}}, \bibinfo {author}
  {\bibfnamefont {T.~G.~F.}\ \bibnamefont {Li}}, \bibinfo {author}
  {\bibfnamefont {D.}~\bibnamefont {Meacher}}, \bibinfo {author} {\bibfnamefont
  {A.}~\bibnamefont {Pace}}, \bibinfo {author} {\bibfnamefont {S.}~\bibnamefont
  {Privitera}}, \bibinfo {author} {\bibfnamefont {L.}~\bibnamefont
  {Sadeghian}}, \bibinfo {author} {\bibfnamefont {L.}~\bibnamefont {Wade}},
  \bibinfo {author} {\bibfnamefont {M.}~\bibnamefont {Wade}}, \bibinfo {author}
  {\bibfnamefont {A.}~\bibnamefont {Weinstein}},\ and\ \bibinfo {author}
  {\bibfnamefont {S.~L.}\ \bibnamefont {Xiao}},\ }\href
  {https://doi.org/10.48550/arXiv.1901.08580} {\bibinfo {title} {The {{GstLAL
  Search Analysis Methods}} for {{Compact Binary Mergers}} in {{Advanced
  LIGO}}'s {{Second}} and {{Advanced Virgo}}'s {{First Observing Runs}}}}
  (\bibinfo {year} {2019}),\ \Eprint {https://arxiv.org/abs/1901.08580}
  {arXiv:1901.08580 [gr-qc]} \BibitemShut {NoStop}%
\bibitem [{\citenamefont {Boland}\ \emph {et~al.}(2021)\citenamefont {Boland},
  \citenamefont {Naus},\ and\ \citenamefont
  {Zwamborn}}]{bolandPhaseResponseReconstruction2021}%
  \BibitemOpen
  \bibfield  {author} {\bibinfo {author} {\bibfnamefont {T.}~\bibnamefont
  {Boland}}, \bibinfo {author} {\bibfnamefont {R.}~\bibnamefont {Naus}},\ and\
  \bibinfo {author} {\bibfnamefont {P.}~\bibnamefont {Zwamborn}},\ }\href
  {https://doi.org/10.1049/smt2.12063} {\bibfield  {journal} {\bibinfo
  {journal} {IET Science, Measurement \& Technology}\ }\textbf {\bibinfo
  {volume} {15}},\ \bibinfo {pages} {619} (\bibinfo {year} {2021})},\ \Eprint
  {https://arxiv.org/abs/2003.13781} {arXiv:2003.13781 [eess]} \BibitemShut
  {NoStop}%
\bibitem [{\citenamefont
  {Nakahara}(2005)}]{nakaharaGeometryTopologyPhysics2005}%
  \BibitemOpen
  \bibfield  {author} {\bibinfo {author} {\bibfnamefont {M.}~\bibnamefont
  {Nakahara}},\ }\href@noop {} {\emph {\bibinfo {title} {Geometry, Topology,
  and Physics}}},\ \bibinfo {edition} {2nd}\ ed.,\ Graduate Student Series in
  Physics\ (\bibinfo  {publisher} {Inst. of Physics Publishing},\ \bibinfo
  {address} {Bristol},\ \bibinfo {year} {2005})\BibitemShut {NoStop}%
\bibitem [{\citenamefont {Amari}\ and\ \citenamefont
  {Nagaoka}(2000)}]{amariMethodsInformationGeometry2000}%
  \BibitemOpen
  \bibfield  {author} {\bibinfo {author} {\bibfnamefont {S.-i.}\ \bibnamefont
  {Amari}}\ and\ \bibinfo {author} {\bibfnamefont {H.}~\bibnamefont
  {Nagaoka}},\ }\href@noop {} {\emph {\bibinfo {title} {Methods of Information
  Geometry}}},\ \bibinfo {series} {Translations of Mathematical Monographs},
  Vol.\ \bibinfo {volume} {191}\ (\bibinfo  {publisher} {American Mathematical
  Society},\ \bibinfo {year} {2000})\BibitemShut {NoStop}%
\bibitem [{\citenamefont {Rao}(1945)}]{raoInformationAccuracyAttainment1945}%
  \BibitemOpen
  \bibfield  {author} {\bibinfo {author} {\bibfnamefont {C.~R.}\ \bibnamefont
  {Rao}},\ }\href@noop {} {\bibfield  {journal} {\bibinfo  {journal} {Bulletin
  of the Calcutta Mathematical Society}\ }\textbf {\bibinfo {volume} {37}},\
  \bibinfo {pages} {81} (\bibinfo {year} {1945})}\BibitemShut {NoStop}%
\end{thebibliography}%

\end{document}